\documentclass[12pt]{article}

\usepackage{graphicx}
\usepackage{subcaption}
\usepackage{amsmath}
\usepackage{dsfont}
\usepackage{placeins}
\usepackage{amssymb}
\usepackage{abstract}
\usepackage[auth-sc,affil-sl]{authblk}
\usepackage{hyperref}
\usepackage{apacite}
\usepackage{enumerate}
\usepackage{fancyhdr}
\usepackage[xcdraw,table]{xcolor}
\usepackage{natbib}
\usepackage{algorithm}
\usepackage{algorithmicx}
\usepackage{algcompatible}
\usepackage{algpseudocode}
\usepackage[table]{xcolor}
\usepackage{multirow}
\usepackage{amsthm}
\usepackage{algpseudocode}
\usepackage{amsthm}
\usepackage{hyperref}
\usepackage{mathrsfs}
\usepackage{amsfonts}
\usepackage{bbding}
\usepackage{listings}
\usepackage{appendix}
\usepackage[margin=1.00in]{geometry}
\usepackage{fancyvrb}
\usepackage{cancel}
\newtheorem{theorem}{Theorem}[section]
\newtheorem{lemma}{Lemma}[theorem]
\newtheorem{corollary}{Corollary}[theorem]

\newtheorem{assumption}{Assumption}[section]

\newcounter{probnum}
\allowdisplaybreaks

\definecolor{tabblue}{rgb}{.870588,.905882,.94902}
\definecolor{gray}{rgb}{0.7,0.7,0.7}
\definecolor{black}{rgb}{0,0,0}
\definecolor{white}{rgb}{1,1,1}
\definecolor{blue}{rgb}{0.0,0.0,1}

\definecolor{green}{rgb}{0,0.5,0}

\definecolor{yellow}{rgb}{1,0.549,0}

\definecolor{red}{rgb}{0.6,0.0,0.0}

\definecolor{darkred}{rgb}{0.9,0.4,0}

\definecolor{purple}{rgb}{0.58,0,0.827}

\definecolor{backgcode}{rgb}{0.97,0.97,0.8}
\definecolor{Brown}{cmyk}{0,0.81,1,0.60}
\definecolor{OliveGreen}{cmyk}{0.64,0,0.95,0.40}
\definecolor{CadetBlue}{cmyk}{0.62,0.57,0.23,0}

\newcommand{\qu}[1]{``{#1}''}
\newcommand{\bv}[1]{\boldsymbol{#1}}

\newcommand{\sigsq}{\sigma^2}

\newcommand{\betaT}{\beta}

\newcommand{\POW}{\mathcal{P}}

\newcommand{\betaThat}{\hat{\beta}}
\newcommand{\betaThatexp}{\hat{\beta}_{\text{run}}}

\newcommand{\Ybar}{\bar{Y}}
\newcommand{\YbarT}{\Ybar_T}
\newcommand{\YbarC}{\Ybar_C}

\newcommand{\allocspace}{\mathbb{W}}
\newcommand{\allocspaceR}{{\allocspace}_{2R}}

\newcommand{\iid}{~{\buildrel iid \over \sim}~}

\newcommand{\Op}[1]{\mathcal{O}_p\parens{#1}}

\newcommand{\half}{\frac{1}{2}}

\newcommand{\B}{\bv{B}}

\newcommand{\M}{\bv{M}}
\newcommand{\m}{\bv{m}}

\newcommand{\x}{\bv{x}}

\newcommand{\w}{\bv{w}}
\newcommand{\wexp}{\w_{\text{run}}}
\newcommand{\yexp}{\y_{\text{run}}}

\newcommand{\quantile}[2]{\text{Quantile}\bracks{#1, #2}}
\newcommand{\onevec}{\bv{1}}
\newcommand{\zerovec}{\bv{0}}

\newcommand{\y}{\bv{y}}

\newcommand{\z}{\bv{z}}

\newcommand{\twovec}[2]{\bracks{\begin{array}{c} #1 \\ #2 \end{array}}}

\newcommand{\twobytwomat}[4]{\bracks{\begin{array}{cc} #1 & #2 \\ #3 & #4 \end{array}}}

\newcommand{\reals}{\mathbb{R}}

\newcommand{\naturals}{\mathbb{N}}

\newcommand{\floor}[1]{\left\lfloor #1 \right\rfloor}

\newcommand{\beqn}{\vspace{-0.25cm}\begin{eqnarray*}}
\newcommand{\eeqn}{\end{eqnarray*}}
\newcommand{\bneqn}{\vspace{-0.25cm}\begin{eqnarray}}
\newcommand{\eneqn}{\end{eqnarray}}
\newcommand{\benum}{\begin{enumerate}}
\newcommand{\eenum}{\end{enumerate}}

\newcommand{\parens}[1]{\left(#1\right)}
\newcommand{\squared}[1]{\parens{#1}^2}

\newcommand{\prob}[1]{\mathbb{P}\parens{#1}}
\newcommand{\probsub}[2]{\mathbb{P}_{#1}\parens{#2}}
\newcommand{\cprob}[2]{\prob{#1~|~#2}}

\newcommand{\bracks}[1]{\left[#1\right]}
\newcommand{\braces}[1]{\left\{#1\right\}}

\newcommand{\abss}[1]{\left|#1\right|}

\newcommand{\expe}[1]{\mathbb{E}\bracks{#1}}

\newcommand{\expesub}[2]{\mathbb{E}_{\,#1}\bracks{#2}}

\newcommand{\indic}[1]{\mathds{1}_{#1}}
\newcommand{\var}[1]{\mathbb{V}\text{ar}\bracks{#1}}
\newcommand{\varsub}[2]{\mathbb{V}\text{ar}_{#1}\bracks{#2}}

\newcommand{\corr}[2]{\text{Corr}\bracks{#1, #2}}

\newcommand{\sesub}[2]{\mathbb{S}\text{E}_{#1}\bracks{#2}}

\newcommand{\oneover}[1]{\frac{1}{#1}}
\newcommand{\overtwo}[1]{\frac{#1}{2}}
\newcommand{\overn}[1]{\frac{#1}{n}}
\newcommand{\oneoversqrt}[1]{\oneover{\sqrt{#1}}}

\newcommand{\mathor}{~~\text{or}~~}

\newcommand{\multnormnot}[3]{\mathcal{N}_{#1}\parens{#2,\,#3}}

\newcommand{\normnot}[2]{\mathcal{N}\parens{#1,\,#2}}

\newcommand{\stdnormnot}{\normnot{0}{1}}

\newcommand{\binomial}[2]{\mathrm{Binomial}\parens{#1,\,#2}}
\newcommand{\hyper}[3]{\mathrm{Hypergeometric}\parens{#1,\,#2,\,#3}}

\newcommand{\zeroonecl}{\bracks{0,1}}

\newcommand{\convd}{~{\buildrel \mathcal{D} \over \longrightarrow}~}

\newcommand{\ourtitle}{Improving the Power of the Randomization Test}
\title{
\ourtitle
}

\author[1]{Abba M. Krieger\thanks{Electronic address: \texttt{krieger@wharton.upenn.edu}; Corresponding author}}
\author[2]{David Azriel\thanks{Electronic address: \texttt{davidazr@technion.ac.il}; Corresponding author}}
\author[3]{Michael Sklar\thanks{Electronic address: \texttt{sklarm@stanford.edu}; Corresponding author}}
\author[4]{Adam Kapelner\thanks{Electronic address: \texttt{kapelner@qc.cuny.edu}; Principal Corresponding author}}

\affil[4]{\small Department of Mathematics, Queens College, CUNY}
\affil[1]{Department of Statistics, The Wharton School of the University of Pennsylvania}
\affil[3]{Department of Statistics, Stanford University}
\affil[2]{Faculty of Industrial Engineering and Management, The Technion, Haifa, Israel}

\begin{document}
\maketitle

\begin{abstract}
We consider the problem of evaluating designs for a two-arm randomized experiment with the criterion being the power of the randomization test for the one-sided null hypothesis. Our evaluation assumes a response that is linear in one observed covariate, an unobserved component and an additive treatment effect where the only randomness comes from the treatment allocations. It is well-known that the power depends on the allocations' imbalance in the observed covariate and this is the reason for the classic restriced designs such as rerandomization. We show that power is also affected by two other design choices: the number of allocations in the design and the degree of linear dependence among the allocations. We prove that the more allocations, the higher the power and the lower the variability in the power. Designs that feature greater independence of allocations are also shown to have higher performance. 

Our theoretical findings and extensive simulation studies imply that the designs with the highest power provide thousands of highly independent allocations that each provide nominal imbalance in the observed covariates. These high powered designs exhibit less randomization than complete randomization and more randomization than recently proposed designs based on numerical optimization. Model choices for a practicing experimenter are rerandomization and greedy pair switching, where both outperform complete randomization and numerical optimization. The tradeoff we find also provides a means to specify the imbalance threshold parameter when rerandomizing.

\end{abstract}
\vspace{2.5cm}
\pagebreak

\section{Introduction}\label{sec:intro}

Our goal is to examine experimental power when testing for a positive treatment effect via the randomization test after a classic treatment-control experiment. The \emph{subjects} have a continuous \emph{response} and each subject is \emph{assigned} to the treatment group ($T$) and a control group ($C$) and subjects' \emph{covariates} are known beforehand and considered fixed. This non-sequential setting was studied by \citet{Fisher1925} when assigning treatments to agricultural plots and is still of great importance today. They occur in clinical trials as \qu{many phase I studies use `banks' of healthy volunteers ... [and] ... in most cluster randomised trials, the clusters are identified before treatment is started} \citep[page 1440]{Senn2013}. 

The practitioner has a choice before the experiment begins that can affect the power: the experimental \emph{design} --- the set of \emph{allocations} drawn from when assigning the subjects to either treatment or control. The naive design is complete randomization (all possible assignments) and well-studied alternatives include rerandomization (a subset of of assignments that satisfy a similarity condition of the covariate values in the two groups) and 1:1 matching.

To investigate experimental power, we formalize our notation, assume a response model and then compute explicit expressions for the power of the randomization test in Section~\ref{sec:formulation}. During this exercise we reveal salient features of the design that are the drivers of power. Section~\ref{sec:results} then investigates power of common designs and provides mathematical results and Section~\ref{sec:simulations}  provides simulations. We conclude with practical advice and future directions in our discussion, Section~\ref{sec:discussion}.

\section{Formulation}\label{sec:formulation}

%In order to explain our designs, there are many prerequisites with background literature that must be explained carefully. The reader who is familiar with the randomization model, the population model and restricted randomization can skip to Section~\ref{subsec:reandomization_design}.

We denote the responses $\y = \bracks{y_1, \ldots, y_n}^\top$ where the number of subjects $n$ is assumed even. The covariate values for all subjects is denoted $\x = \bracks{x_1, \ldots, x_n}^\top$ which is assumed centered and scaled. The \emph{assignment} or \emph{allocation} vector is $\w = \bracks{w_1, \ldots, w_n}^\top$ whose entries are either +1 (the subject received $T$) or -1 (the subject received a $C$) and $\w \in \braces{-1,+1}^n$. The \emph{design} $D$ we define as a discrete uniform random variable with support $\allocspace_D \subseteq \braces{-1,+1}^n$. The subset of designs we investigate are termed \emph{forced balance procedures} where all allocations have the same number of treated and control subjects \citep[Chapter 3.3]{Rosenberger2016}. This is a minor restriction denoted as $\allocspace_{FB} = \braces{\w\,:\, \w^\top \onevec_n = 0}$ which has $\binom{n}{n/2}$ allocations. Herein, every design is distinguished by its specific $\allocspace_D \subseteq \allocspace_{FB}$. We further assume all assignments in all $D$ divide the subjects into two subsets of equal size. We further assume that the subset assigned to treatment is chosen with probability 1/2. We call this the mirror property of our designs $D$. 

We will examine a few designs in this paper: the balanced complete randomization design (BCRD) is specified by $\allocspace_{FB}$, rerandomization design by $\braces{\w\,:\, \abss{\w^\top \x} \leq a} \cap \allocspace_{FB}$ where $a$ is a threshold of covariate imbalance (to be elaborated upon later), pairwise matching design by $\braces{\w\,:\,w_r = -w_s}$ for unordered pairs with indices $\braces{r,s}$ in the set of binary matches (which is naturally a subset of $\allocspace_{FB}$) and the greedy pair-switching design of \citet{Krieger2019} which is a subset of $\allocspace_{FB}$ with covariate imbalance provably lower than the previous designs.

We assume the following response model

\bneqn\label{eq:simple_model}
\y = \betaT \w + \beta_x \x + \z
\eneqn

\noindent where $\z$ is the unexplained but fixed component after an additive treatment effect and a linear covariate effect. The only source of the randomness in the response is thus the treatment assignments $\w$. This assumption on the source of randomness is terms the \emph{randomization model} \citep[Chapter 6.3]{Rosenberger2016}, the \qu{Fisher model} or the \qu{Neyman model} whereby \qu{the $n$ subjects are the population of interest; they are not assumed to be randomly drawn from a superpopulation} \citep[page 297]{Lin2013}.

Our focus is to assess $H_a: \betaT > 0$ can be shown beyond a reasonable doubt over $H_0: \betaT \leq 0$ at level $\alpha$. We employ the simple \emph{differences-in-means} estimator, 

\bneqn\label{eq:estimator}
\betaThat := \overn{\w^\top \y} = \half(\YbarT - \YbarC)
\eneqn 

\noindent where the equality follows from our assumption of forced balance (i.e. allocations have the same number of treated and control subjects). Our criterion to evaluate experimental designs is the statistical power of the randomization test at level $\alpha$.

To \qu{run an experiment}, a $\wexp$ is first chosen at random from $\allocspace_D$ whose size is denoted $R_D := \abss{\allocspace_D}$. The allocation $\wexp$ is used to generate the responses $\yexp$ in Equation~\ref{eq:simple_model} and the estimate of the treatment effect $\betaThatexp$. Our previous work \citep{Kapelner2020} examined the effect of $D$ on the mean squared error of this estimator in the same settings.

The decision of the randomization test compares this estimate to the \emph{null distribution}, $\betaThat~|~H_0 \sim \text{Unif}\parens{\braces{\w^\top \yexp / n\,:\, \w \in \allocspace_D}}$, and if larger than the $1-\alpha$ quantile, $H_0$ is rejected. The power of the randomization test considers a universe of each $\wexp$ composing $R_D$ experiments with $R_D$ different experimental responses. 

In practice, since $\allocspace_D$ is exponentially large for the designs we consider, we can use an approximation found in \citet[page 636]{Lehmann2006}. We randomly choose a unique subset $\allocspaceR := \braces{\w_1, -\w_1, \w_2, -\w_2, \ldots, \w_{R}, -\w_{R}} \subset \allocspace_D$ that has $R$ mirrored pairs and thus a total of $2R$ vectors, an approximation method that dates back to \citet{Dwass1957} for permutation testing. (For common designs, $2R \ll R_D$ but for some designs whose allocations satisfy an \qu{optimality} criterion it may not be so). In the discussion that follows, we employ these $2R$ vectors to make an approximate experimental decision and approximate power computation. $R$ is a choice of the experimenter limited only by computational constraints. We demonstrate later that its choice is pivotal to power. 

Consider $\M$, the $2R \times 2R$ matrix of estimates where the row $i$ indexes the experimental run allocation, i.e. $\wexp = \w_i$, and the column $j$ indexes the allocation corresponding to the $j$th element of its null distribution. For convenience, the rows of $\M$ are organized by the mirrored couples i.e. the first row corresponds to $\w_1$ and the second row corresponds to $-\w_1$ and the third row corresponds to $\w_2$, etc. The diagonal elements are the estimators of the $\betaThatexp$ for each $\wexp$.

For any run $\w_i$, the null hypothesis is rejected if the element in the $i$th column is greater than $\quantile{\m_{i\cdot}}{1 - \alpha}$ where $\m_{i\cdot}$ denotes the $i$th row vector of $\M$. We examine approximate experimental power, denoted $\POW_{\z, \allocspaceR}$, the proportion of rejections over our subset of $2R \ll R_D$ run allocations,

\bneqn\label{eq:pow}
\POW_{\z, \allocspaceR} := \oneover{2R} \sum_{i=1}^{2R} \indic{m_{i,i} > \quantile{\m_{i\cdot}}{1 - \alpha}} 
\eneqn

\noindent where $\indic{\cdot}$ denotes the indicator function, $m_{i,j}$ denotes the $i,j$ element of $\M$ and the conditioning on $\z$ emphasizes the dependence on the unobserved component of the response. The exact power would be the expectation over the entire allocations space $\allocspace_D$. This computation is seldom possible because for most designs, $R_D$ is practically infinite at even modest sample sizes. 
 
To understand how these power expressions are dependent on the choice of design, we express the matrix entries as:

\bneqn\label{eq:m_entry}
m_{ij} = \beta r_{ij}  + \beta_x B_{x,j} + B_{z,j}
\eneqn

\noindent where $r_{ij} := \w_i^\top \w_j / n$, $B_{x,j} := \w_j^\top \x / n$ and $B_{z,j} := \w_j^\top \z / n$, termed pairwise allocation correlation, imbalance in the observed covariate and imbalance in the unobserved response component respectively. Note that $\wexp$ affects only the first term above.

The quantity $\w_i^\top \w_j$ is the number of allocations among two designs that agree save the number of allocations that disagree and thus $r_{ij} = 2f_{ij} - 1$ where $f_{ij}$ is the fraction of allocations that agree in the two designs. Also $r_{ij}$ varies between $\bracks{-1, +1}$ and $r_{ii} = 1$.

In each experimental run (each row of $\M$), we compare null estimates with the data-generated estimate $m_{ii}$ and desire $m_{ii} > m_{ij}$ for most $j$ to reject. If this is true for a large proportion of the $R$ experimental runs, we reject the null often resulting in high power. Hence the following difference quantity being positive with high probability is of great importance:

\bneqn\label{eq:betahat_diff}
m_{ii} - m_{ij} &=& (\beta + \beta_x B_{x,i} + B_{z,i}) - (r_{ij} \beta + \beta_x B_{x,j} + B_{z,j}) \nonumber \\
&=& \underbrace{(B_{x,i} - B_{x,j})}_{\text{I}} +  \underbrace{(B_{z,i} - B_{z,j})}_{\text{II}} +\underbrace{\beta (1 - r_{ij})}_{\text{III}}
\eneqn

\noindent This shows that power fundamentally depends on (I) the imbalance in the observed covariates, (II) the imbalance in the unobserved response component and (III) the correlations among the set of allocations. All terms are functions of $\w$ making them all critically dependent on the choice of experimental design. The experimental settings and response model (i.e. $n$, $\beta$ and $\beta_x$ and $\sigsq_z$) affect power: $\beta_x$ relative to $\sigsq_z$ (similar to $R$-squared) will determine how much the observed covariate balance (term I) will affect the power, the sample size $n$ has well-known asymptotics and $\beta$ increasing will always increase power. But how can power be maximized by the experimentalist through $D$? 

First note that III is always positive and I and II are positive for half the allocations $i$ and negative for the other half of the allocations (corresponding to the mirrored allocations). This is why, as a general strategy, using designs that make terms (I) and (II) close to zero in absolute value maximizes the power of the randomization test. 

%This takes some explanation. One way of making this large is to let $B_{x,i}$ and $B_{z,i}$ be large and let $B_{x,j}$ and $B_{z,j}$ to be small. However, this would only work for one row $i$. In the imrrored design, the terms I and II would then be negative.

%%%%%%

%\inred{I know this is of \qu{great importance}, but how exactly is it connected to power? We provide a lot of $O(\cdot)$ statements below. How is the order in $n$ connected to power?}\\

In term (I), since $\x$ is observed, the design can be tailored to make $B_x$ small in absolute value for all allocations. If BCRD is employed, (I) $= \Op{n^{-1/2}}$, pairwise matching design yields (I) $= \Op{n^{-1}}$, the randomization design makes (I) small through choosing threshold $a$ to be small and greedy pair-switching has (I) $= \Op{n^{-3}}$.

In term (II), since $\z$ is not observed, the design cannot be tailored to make $B_z$ small in absolute value for any allocation. One may have thought that the design does not affect this term, but we see above and in our previous work \citep{Kapelner2020} that surprisingly it design does affect what is unseen. The key is the relatonship between $B_{z,i}$ and $B_{z,j}$; their correlation is

%$B_{z,i} = \sum_{k=1}^n w_{ik} z_{k} / n$ and $B_{z,j} = \sum_{k=1}^n w_{jk} z_{k} / n$, 
\bneqn\label{eq:corrBziBzj}
\corr{B_{z,i}}{B_{z,j}} = r_{ij}.
\eneqn

\noindent Thus, independent of design, term (II) has variance $2\sigsq_z(1 - r_{ij}) / n$ where $\sigsq_z$ is the population variance of $\z$. Although it is always $\Op{n^{-1}}$, their is a crucial constant that is determined solely by our choice of design: the smaller the $r_{ij}$ values within the elements of $\allocspaceR$, the higher the power.

The order of term (III) depends on the behavior of $r_{ij}$ and since the expected $r_{ij}$ is zero due to our mirror assumption, this order will be the standard deviation of $r_{ij}$. BCRD features $r_{ij} = (2U-1)/n$ where $U \sim \hyper{n}{n/2}{1/2}$. Since $(n/2-1)/(n-1) \approx 1/2$, the standard deviation over allocations of $r_{ij}$ is  approximately $1/\sqrt{4n}$ and greedy pair-switching is the same. Pairwise matching design has $U \sim \binomial{n/2}{1/2}$ and thus the standard deviation over allocations of $r_{ij}$ is exactly $1 / \sqrt{2n}$. The behavior of the standard deviation of $r_{ij}$ in the rerandomization design is not known.

In addition to the imbalance in $\x$, the choice of design affects power through the pairwise allocation correlations $r_{ij}$'s but this relationship is complicated. In term (II) power is increased with large pairwise allocation correlations and in term (III) power is increased by making small pairwise correlations. Theoretical results in the next section show that term (III) dominates and thus designs that feature pairwise correlations as small as possible are preferred. The next section also shows the prominent role of the number of allocations in the design $R$ which is not apparent from this analysis.

\section{Results}\label{sec:results}

As shown in the previous section, the experimentalist's design decision affects power of the randomization test through (1) the number of unmirrored allocations $R$ and (2) the dependence of the allocations within the designs, $r_{ij}$'s. The goal of this section is to prove theoretical results about the behavior of these two parameters. We begin by rewriting Equation~\ref{eq:pow} as

\beqn
\POW_{\z, \allocspaceR} = \oneover{R} \sum_{i=1}^{R} (I_i + I_i^m)
\eeqn

\noindent where $I_i$ denotes the indicator that the treatment effect estimator when $\wexp = \w_i$ is higher than all but the top $1- \alpha$ proportion of estimates from the other $2R - 1$ allocations (i.e. $H_0$ is rejected) and $I_i^m$ is the analogous indicator for the mirror allocation, $\wexp = -\w_i$. We now define our power metric, which averages the above expression over $B_z$ and all subsets $\allocspaceR$ of the full allocation space,
 
\bneqn\label{eq:def_power_metric}
\POW := \expesub{B_z, \allocspaceR}{\POW_{\z, \allocspaceR} } = \oneover{2R} \sum_{i=1}^{R} (\prob{I_i = 1} + \prob{I_i^m = 1}) = \overtwo{\prob{I_1 = 1} + \prob{I_1^m = 1}}
\eneqn

\noindent where the last equality follows without loss of generality. The difficulty in determining $\prob{I_1 = 1}$ is that the computation involves order statistics of correlated random variables. To even feasibly compute $\prob{I_1 = 1}$, we must make a few simplifying assumptions:

\begin{assumption}[Normality of the Imbalances in the Unseen Response Component]\label{ass:Bznorm}~\\
$B_{z,j} := \sum_{\ell=1}^n w_{j,\ell} z_{\ell} / n \sim \normnot{0}{\sigsq_z / n}$.
\end{assumption}

\begin{assumption}[Uniformity of the absolute allocation correlations]\label{ass:all_rho}~\\
$\rho := \abss{r_{ij}}$ for all allocations $i \neq j$ where $i,j = 1, \ldots, R$.
\end{assumption}

\begin{assumption}[Trivial observed covariate imbalance]\label{ass:noBx}~\\
$B_{x,i} = 0$ for all allocations $i = 1, \ldots, 2R$.
\end{assumption}

Assumption~\ref{ass:Bznorm} involves an extension of the finite central limit theorem from \citet{Li2017} who follows \citet{Hajek1961} proving that for BCRD, $\sqrt{n} B_z \convd \normnot{0}{\sigsq_z}$. This simplifying assumption is reasonable in this context given a large number of allocation vectors and sufficient sample size. This assumption removes the dependence of power on the $n$ fixed values of the $z_\ell$'s shifting the dependence to more tractable realizations from a normal model. Assumption~\ref{ass:all_rho} simplifies the dependence of power on of order $R^2$ number of correlation parameters to just one parameter $\rho$. It also simplifies our power metric (Equation~\ref{eq:def_power_metric}) where the expectation is no longer taken over different subsets $\allocspaceR$ as this assumption implies each allocation vector subset is equivalent in this context. Assumption~\ref{ass:noBx} is effectively true in all designs we consider save BCRD. Also, for convenience we let $\sigsq_z = 1$ which do not change any of our theoretical results that follow. These assumptions may seem strong, but simulation results in Section~\ref{sec:simulations} comport with our theoretical results in this section; thus, these assumptions do not seem to be restrictive.

These assumptions allow us to make progress on computing the probabilities $\prob{I_1 = 1}$ and $\prob{I_1^m = 1}$. To do so, we must make comparisons of the estimator under $\wexp = \w_1$ to the other $2R - 1$ estimators under other allocations. To emphasize that these are random variables, we denote these estimators as $V_{1,1}, V_{1,1}^m, V_{1,2}, V_{1,2}^m, \ldots, V_{1,R}, V_{1,R}^m$ (which are the random variables in row $\m_{1 \cdot}$ computed under the three assumptions). Multiplying through by $\sqrt{n}$, letting $\gamma := \sqrt{n} \beta$ implies that $\sqrt{n} \B_z$ is an equicorrelated standard multivariate normal where $\B_z$ is the vector of $B_{z_j}$'s for unmirrored allocations $j = 1, \ldots, R$. Hence we can write the estimators as

\beqn
V_{11} &=& \sqrt{\rho}Z_0 + \sqrt{1 - \rho} Z_1  + \gamma,\\
V_{11}^m &=& -V_{11}, \\
V_{1j} &=& \sqrt{\rho}Z_0 + \sqrt{1 - \rho} Z_j  + \rho \gamma,\\
V_{1j}^m &=& -V_{1j}.
\eeqn

\noindent where $Z_0, Z_1, \ldots, Z_R \iid \stdnormnot$. To compute

\bneqn\label{eq:prob_indicator}
\prob{I_1 = 1} = \prob{V_{1,1} > \quantile{\braces{V_{1,1}, V_{1,1}^m, V_{1,2}, V_{1,2}^m, \ldots, V_{1,R}, V_{1,R}^m}}{1 - \alpha}}
\eneqn

\noindent we first examine the case where $R \rightarrow \infty$. Here, the quantile in Equation~\ref{eq:prob_indicator} above is fixed if we condition on $Z_0 = z$ and can be found by solving

\bneqn\label{eq:qz_R_asmyptotic}
\Phi\parens{\frac{\sqrt{\rho}z+\rho \gamma-q(z)}{{\sqrt{1-\rho}}}}+ 
\Phi\parens{-\frac{\sqrt{\rho}z+\rho \gamma+q(z)}{{\sqrt{1-\rho}}}} = 2\alpha
\eneqn

\noindent for $q(z)$ where $\Phi(\cdot)$ computes the value of the CDF of the standard normal distribution. Then the probability of $\prob{I_1 = 1}$, which here is the same as the value of $\prob{I_1^m = 1}$, reduces to a normal CDF calculation which is a function of $z$. Averaging over the distribution of $Z_0$ we obtain

\bneqn\label{eq:power_asymptotic_in_R}
\lim_{R \rightarrow \infty} \POW = \int_\reals \Phi\parens{\frac{\gamma + \sqrt{\rho} z - q(z)}{\sqrt{1 - \rho}}} \phi(z)dz
\eneqn

\noindent where $\phi(\cdot)$ computes the value of the PDF of the standard normal distribution. Details can be found in Theorem~\ref{thm:power_asymptotic_in_R} in the Appendix.

We now examine the case of finite $R$. Here, we need to compare $V_{11}$ to every other estimator and count the number of times $V_{11}$ is larger. We change variables from $(Z_0, Z_1)$ to a rotation $(U, S)$ which greatly simplifies our analysis (see Equation~\ref{eq:u_and_s_change_of_variables} in the Appendix). When we condition on $U = u$ and $S = s$, the events $\abss{V_{1j}} > V_{11}$ become iid Bernoulli random variables whose probability parameter is a function of $u$ and $s$. Then $\cprob{I_1 = 1}{U = u, S= s}$ reduces to a binomial CDF calculation. We then average the computation over the distribution of both $U$ and $S$ to arrive at

\bneqn\label{eq:power_R}
\POW = \int_0^\infty \int_\reals F_{B}( \floor{2\alpha R}-1; R-1, p(u, s)) \phi(s; 0 ,\rho) \phi\parens{u; \frac{\gamma}{\sqrt{1 - \rho}}, \frac{1}{1 - \rho}} dsdu
\eneqn

\noindent where $F_B(\cdot, N, \theta)$ is the CDF of the $\binomial{N}{\theta}$ random variable, $\phi(\cdot; \mu, \sigsq)$ computes the value of the PDF of a $\normnot{\mu}{\sigsq}$ random variable and

\beqn
p(u, s) = \Phi\parens{-\parens{1-\rho} u + s} + \Phi\parens{-\parens{1+\rho}u - s}.
\eeqn

\noindent Details can be found in Theorem~\ref{thm:power_fixed_R_expression} in the Appendix. Note that power is discretized at natural number values of $2 \alpha R$ which we term \qu{attainable power} values similar to the concept of attainable $p$-values \citep{Hemerik2019}.

We would like to return to our objective which is to understand the role that $R$ and $\rho$ play in power. To gain intuition about how $R$ and $\rho$ affect the power expression of Equation~\ref{eq:power_R}, we use Monte Carlo integration to compute $\POW$ under the following settings: $R \in \{10, 30, 100, 320, 1000, 3160\} \approx 10^{\braces{1.0, 1.5, \ldots, 3.5}}$, $\rho \in \braces{0.0, 0.1, 0.2, 0.3}$ and $n \in \braces{26, 50, 100, 200}$ fixing $\alpha = 5\%$. The results are illustrated in Figure~\ref{fig:power_conv} along with the calculation of $\lim_{R \rightarrow \infty} \POW$ via Equation~\ref{eq:power_asymptotic_in_R}.

\begin{figure}[h]
\centering
\includegraphics[width=6.4in]{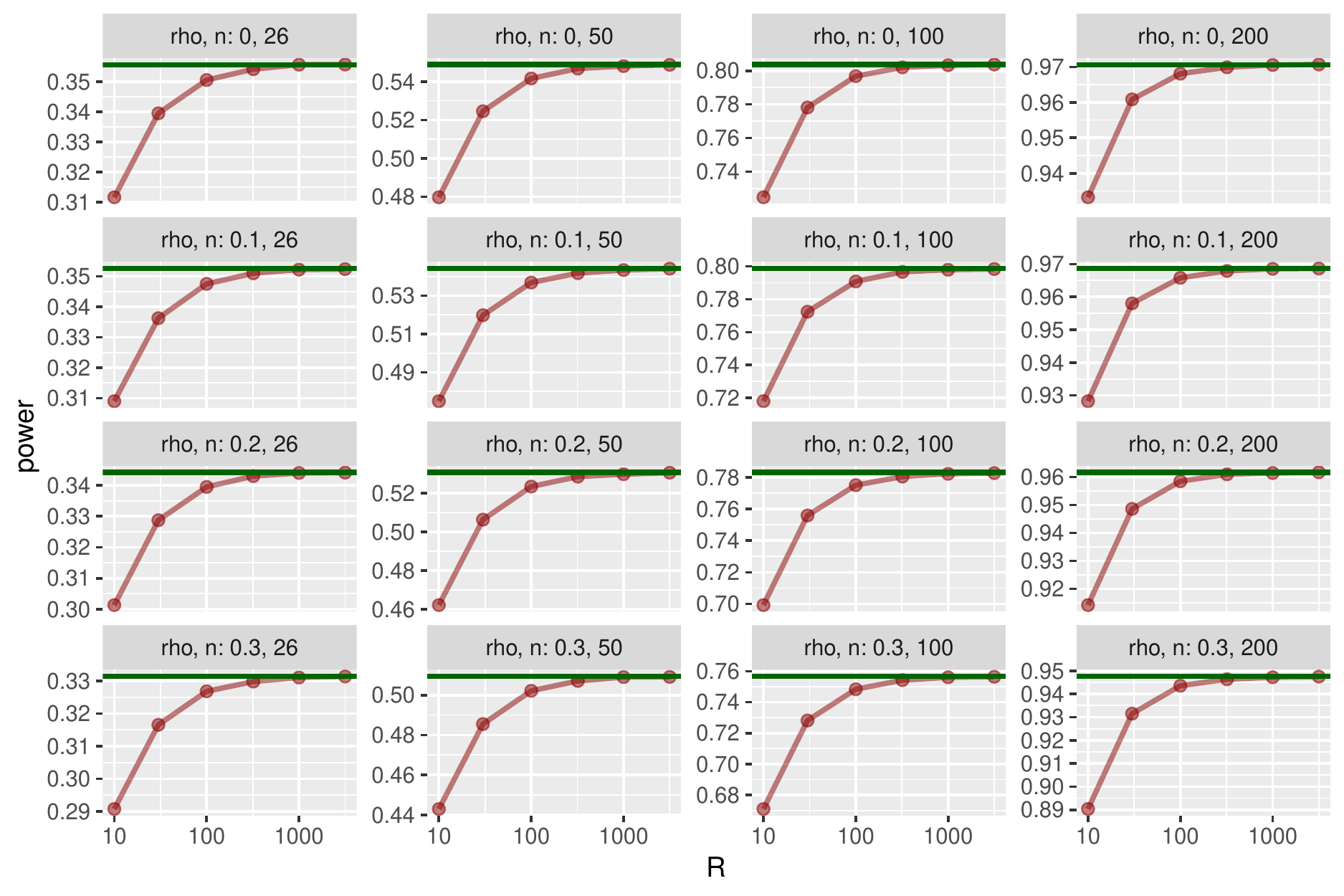}
\caption{An illustration of asymptotic power in $R$ for different fixed settings of $n$ and $\rho$ where $\alpha = 5\%$. The green horizontal line is the theoretical asymptotic power upper bound (Equation~\ref{eq:power_asymptotic_in_R}).  In red is the finite power calculated using the derived expression for power as a function of $R, \rho, \alpha$ (Equation~\ref{eq:power_R}) via Monte Carlo integration (one million samples per cell). The points indicate the values of $R$ for which finite power was computed.}
\label{fig:power_conv}
\end{figure}

Besides confirming the asymptotic power expression of Equation~\ref{eq:power_asymptotic_in_R}, these illustrations beg two conjectures (1) for any fixed $\rho$, power monotonically increases as $R$ increases and (2) for any fixed $R$, power monotonically increases as $\rho$ decreases. 

We first examine conjecture (2). Note that $\rho$ has two competing effects which complicate its overall behavior in power. High values of $\rho$ are preferable because it reduces the variance of $V_{1j} - V_{11}$. Simultaneously, low values of $\rho$ are also preferable because it reduces $V_{1j}$ by diminishing its term $\rho\gamma$. To examine which consideration is stronger, we first examine the case where $R = 2$. Through much manipulation, we prove that power increases as $\rho$ decreases (see Theorem~\ref{thm:toy_case} in the Appendix). 

To understand the role of $\rho$ when $R > 2$, we need to examine Equation~\ref{eq:power_R} which requires understanding the behavior of the random variable $T = p(U, S)$ in the binomial CDF calculation. Assuming $\gamma = 0$, Lemma~\ref{lemm:uniform_for_general_densities} shows that the measure of $T$ has a point mass $\prob{T = 1} = 1/2$ and otherwise has uniform density on $\zeroonecl$ (and this result holds beyond the normality assumption). If $\gamma > 0$, Lemma~\ref{lemm:main_support} shows $\prob{T = 1} = \Phi(-\gamma)$ and the density on $\zeroonecl$ is given by

\beqn
f_T(t; \rho, \gamma) =e^{-\frac{1}{2}\gamma^2} \int_0^\infty e^{\gamma\sqrt{1-\rho} u(a,t) }  \phi\parens{a; \, 0, \textstyle\frac{\rho}{1-\rho}} da
\eeqn

\noindent where $u = u(a, t)$ satisfies 

\beqn
\Phi(-u + a) + \Phi(-u - a) = t.
\eeqn
 
To prove conjecture (2), it is sufficient to demonstrate that there exists a $t_0$ where for $t < t_0$, $f_T(t; \rho_1, \gamma) > f_T(t; \rho_2, \gamma)$ and for $t > t_0$, $f_T(t; \rho_1, \gamma) < f_T(t; \rho_2, \gamma)$ for $\rho_1 < \rho_2$. From numerical integration for a fine grid of $t$ for many pairs of $\rho_1 < \rho_2$, $t_0$ is clearly visible. But this has so far escaped formal proof.

We now turn to Conjecture (1). It is apparent from Figure~\ref{fig:power_conv} that the role of $R$ is more salient than $\rho$. We prove this conjecture in Theorem~\ref{thm:power_increases_monotonically_as_R_increases} by first showing that $f_T$ decreases in $t$ for any $\rho$ and $\gamma > 0$ in Corollary~\ref{corr:density_negative}. We then use this fact in conjunction with the behavior of $f_T$ when $\gamma = 0$ (Lemma~\ref{lemm:uniform_for_general_densities}) and representing the CDF of the binomial as a regularized incomplete beta function allows us to complete the proof.

A further consideration is the variability of the power. Recall the power metric we focus on (Equation~\ref{eq:def_power_metric}) considers power for an average value of the imbalance of the unobserved response component $B_z$. There is variance in $B_z$ which causes instability in experimental power, $\sesub{B_z, \allocspaceR}{\POW_{\z, \allocspaceR}}$. Replacing the value of the quantile in Equation~\ref{eq:prob_indicator} by $q(z)$ from Equation~\ref{eq:qz_R_asmyptotic}, Theorem~\ref{thm:se_pow_decreases} proves that as $R$ increases, this instability monotonically decreases to a positive constant. This limiting constant is a function of $\gamma$ and $\rho$. Numerical studies show that it does not exhibit monotonicity in either of these two parameters. Also, for typical values of $\gamma$ and $\rho$, this limiting constant is large; it could be as high as $0.06$. Future work will elucidate designs that seek to minimize this value.

\section{Simulations}\label{sec:simulations}

In this section, we wish to explore the power $\POW$ of the randomization test at $\alpha = 0.05$ for different experimental design strategies and different values of $R$ and $n$. We vary $n \in \braces{26, 50, 100, 200}$ and then set the observed covariates $\x$ to be the $\{1 / (n + 1), 2 / (n + 1), \ldots,$ $n / (n+1)\}$ quantiles of the standard normal distribution. The designs considered were BCRD, rerandomization, a priori pairwise matching, the greedy pair switching of \citet{Krieger2019}. For rerandomization, we picked a threshold of $\abss{B_x}$ corresponding to the 0.1\% best out of 1,000,000 allocations from BCRD. All these designs are restricted to have an equal number of subjects $n/2$ assigned to the treatment and control groups by construction. 

%The optimal design is defined as first calculating $B_x$ for all $R_D$ vectors, sorting the vectors from smallest $B_x$ to largest and then retaining only the first $R$ vectors. Since $R_D$ is exponentially large, the optimal design in our simulations was only available when $n=26$ (where $R_D = 10,400,600$, a managable number). 

We vary $R \in \{10, 30, 100, 320, 1000, 3160\}$. After sampling the $R$ vectors $\w$, we augment this set by concatenating their mirrors $-\w$ thus arriving at $2R$ vectors for each design. To generate the response we vary $\beta_x \in \braces{0, 1}$ and $\betaT \in \braces{0, 0.25}$. The positive value of $\betaT$ was selected to both induce separation among the many simulation settings and result in powers that were not close to either zero or one. In each simulation cell, we run 50 realizations from each of the considered designs i.e. 50 different subsets $\allocspaceR \subset \allocspace_D$. In order to simulate different values of $B_z$ under the asymptotic setting, we take 500 draws of the unobserved covariates $\z$ from a standard normal distribution within each design duplicate (a different value of $\sigsq_z$ would only monotonically shift our results). The responses were computed according to our theoretical setup (Equation~\ref{eq:simple_model}) and the power in each cell $\POW_{\z, \allocspaceR}$ was tabulated via
Equation~\ref{eq:pow}.

For $n=26$, the total number of allocations is $\binom{26}{13} = 10,400,600$ which is nearly the largest sample size that can comfortably be enumerated exhaustively. Thus, for the $n=26$ setting we include another experimental design, which we term \emph{best}. Here, we calculate the observed covariate imbalance $B_x$ for each allocation vector $\w$. We then sort from the smallest $\abss{B_x}$ to the largest and enumerate the best $2R$ vectors (since the $\abss{B_x}$ is equal for $\w$ and $-\w$, the mirrored pairs appear in order after sorting). Additionally, since the subset of $R$ vectors is deterministic, we do not do 50 duplicates of this design during the simulation.

The main goal of this simulation is to compare power $\POW$ across the various design strategies, $R$ and $n$, computed via Equation~\ref{eq:def_power_metric} where the expectation over $B_z$ was approximated by averaging the 500 replicates over different $\z$ realizations and the expectation over all subsets $\allocspaceR \subset \allocspace_D$ was approximated by averaging over the 50 replicates of the different $\allocspaceR$ subsets. Using the law of total variance, we can compute standard errors of our simulation that incorporate these two sources of variation. 

We also collected other information during the simulation such as an estimate of mean $\rho$, a metric of how similar the allocation vectors are within specific experimental designs by $R$ and $n$ (as measured by the average $\abss{r_{ij}}$). This allows us to understand the interplay of $R$ and $\rho$ on power and assess our theoretical results in settings outside of their stylized assumptions. The power results are found in Figure~\ref{fig:average_power_sims} and the allocation dependence illustrations are found in Figure~\ref{fig:average_abs_rijs}.

%where we showed that $\POW$ increases as $R$ grows (Theorem~\ref{thm:power_increases_monotonically_as_R_increases}) and as $\rho$ shrinks (Theorem~\ref{thm:power_increases_monotonically_as_rho_decreases}). 

\begin{figure}[h]
\centering
\includegraphics[width=6.4in]{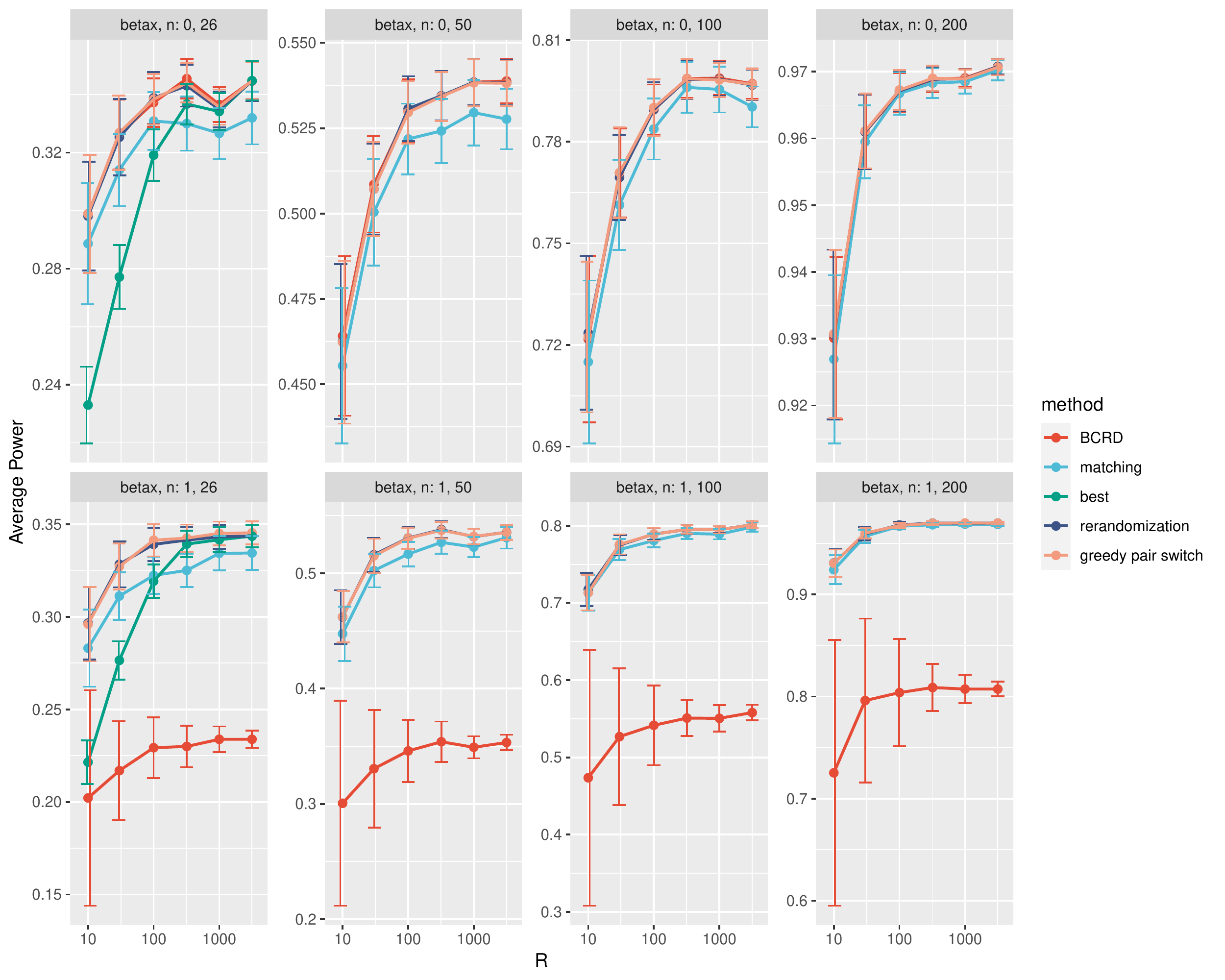}
\caption{Simulated estimates of power $\POW$ of the randomization test by number of allocation vectors in the design $R$ where $\beta = 0.25$, $\alpha = 0.05$ and $\beta_x \in \braces{0, 1}$. Individual plots correspond to different settings of the effect of the observed covariate $\beta_x$ and sample size $n$. Colors indicate the design strategy employed. Error bars are jittered slightly left-right and indicate 95\% confidence.}
\label{fig:average_power_sims}
\end{figure}
%\FloatBarrier

There are many observations from these plots. First, all eight illustrations of Figure~\ref{fig:average_power_sims} confirm our technical result that $\POW$ increases monotonically in $R$ (Theorem~\ref{thm:power_increases_monotonically_as_R_increases}).

Further, the bottom row illustrations of Figure~\ref{fig:average_power_sims} (where $\beta_x \neq 0$) demonstrate the effect of the design on the observed covariate imbalance $B_x$ affecting term (I) in $\betaThat$ (Equation~\ref{eq:betahat_diff}). BCRD has poor balancing performance and hence much lower power than the contenders (see Figure~\ref{fig:average_abs_Bxs} in the supplementary materials for observed covariate imbalance by design and $n$). Matching has worse power compared to rerandomization and greedy pair switching for small sample sizes but is no longer detectable at the illustration scale for $n = 200$. This is due to a combination of worse performance balancing the observed covariate (see Figure~\ref{fig:average_abs_Bxs}) and also higher allocation dependence (as apparent in Figure~\ref{fig:average_abs_rijs}), i.e. higher average $\abss{r_{ij}}$ with the latter consideration being more responsible (as evidenced by the comparison of the top row of Figure~\ref{fig:average_power_sims} corresponding to the setting of $\beta_x = 0$). The strategy of using the best vectors in an exhaustive search does poorly for low $R$ since the vectors are highly dependent as apparent from the leftmost plot of Figure~\ref{fig:average_abs_rijs}.

\begin{figure}[h]
\centering
\includegraphics[width=6.4in]{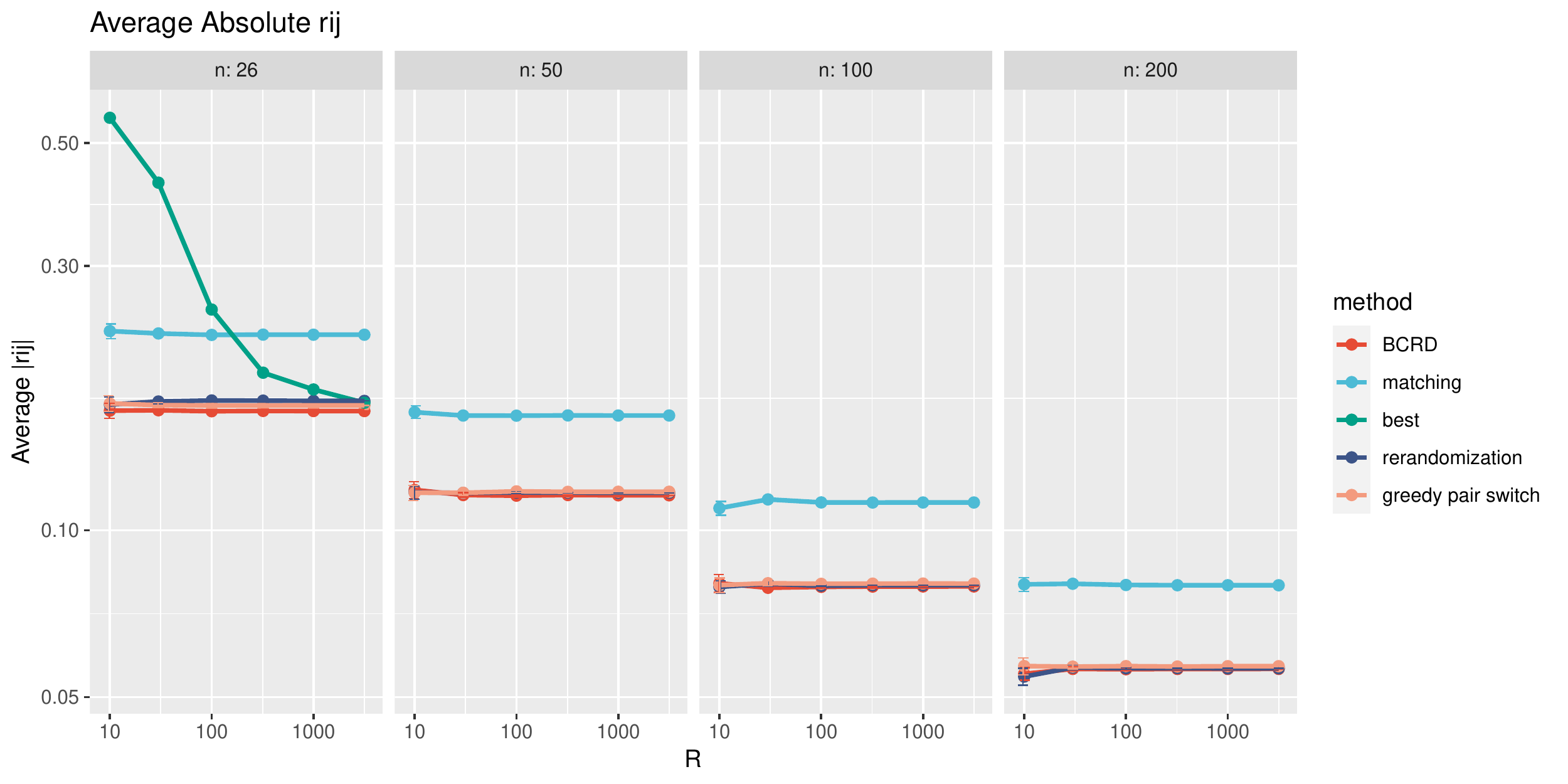}
\caption{Simulated average $\abss{r_{ij}}$ over all $2R$ allocation vectors on the log scale for each design method considered by subset size $R$. Individual plots correspond to different settings the sample size $n$. Allocation vector correlations are independent of $\betaT, \beta_x$ and $\alpha$. Colors indicate the design strategy employed. Error bars are jittered slightly left-right and indicate 95\% confidence but are usually smaller than the dot width. The advantage of pairwise matching is clearly illustrated as well as the equivalence of BCRD, rerandomization at the threshold considered and greedy pair switching.}
\label{fig:average_abs_rijs}
\end{figure}

The winning strategies in the realistic setting of $\beta_x \neq 0$ are rerandomization and greedy pair switching because both these designs can (a) drive $\abss{B_x}$ to nominal levels and (b) provide highly independent allocations.

We also verified that the tests are properly sized in all settings (see Figure~\ref{fig:size_by_method} in the supplementary materials). And we verified that the standard error of $\POW$ monotonically decreases in $R$ but not ultimately to zero. This can be seen from inspecting the length of the error bars of Figure~\ref{fig:average_power_sims} but can be seen more clearly in Figure~\ref{fig:ses_by_method} in the supplementary materials. Again, the winners here are rerandomization and greedy pair switching. The variability in BCRD in small $R$ in the case where $\beta_x \neq 0$ is the largest variability of the simulation. This is likely due to $B_x$ varying wildly across allocations.

\section{Discussion}\label{sec:discussion}

We investigated the power of the randomization test under different experimental designs in the setting of a simple response model with an additive treatment effect, an additive effect of an observed covariate and an additive effect of a fixed unobserved component. We first observe that if the sample size is large, e.g. greater than 200, using a specialized design different from BCRD does not provide significant power gains because both the observed covariate and unobserved response component will have low imbalance, a point noted by \citet{Cornfield1959}. 

But in many experimental settings, for example clinical trials, sample size is relatively small (e.g. $n < 100$) and the experimenter's design choice does indeed matter. Through our investigation of the power, we noticed that three main features of the experimental design (i.e. within the experimentalist's control) affect power: (1) The allocations' imbalance among the observed covariates, (2) the number of allocations in the design and (3) the orthogonality of the allocations.

%and we isolated two main distinct components of how design affects both the $\POW$ and variability in power, imbalance in the observed covariates and orthogonality of allocations.

The most important design consideration is (1) to make imbalance among the observed covariate in the two arms small enough so that it becomes inconsequential. This fact is not new; it is a well-studied problem with many heuristic designs including \citeauthor{Student1938}'s (\citeyear{Student1938}) classic rerandomization and recently \citet{Bertsimas2015,Kallus2018} who employ numerical optimization. 

%One can use rerandomization (which Student recommended in the 1930's), a priori pairwise matching, numerical optimization, etc. We have shown herein that the other two considerations are of import.

For consideration (2), we have demonstrated that it is critical to have a large number of allocations $R$ in the design as $\POW$ will increase and variability of power due to the effect of the unobserved response component will decrease. An order of $R$ in the low 1000's seems to be sufficient. One must be careful that all these many $R$ allocations respect the small observed covarate imbalance restriction. This restriction is explicit in rerandomization where the minimum threshold imbalance is specified. One can prove that imbalance is very small in matching and greedy pair switching or alter these design strategies by requiring an explicit threshold. 

As for consideration (3), we conjecture that power will improve upon increasing the orthogonality among the allocations. Intuitively, (2) and (3) are at odds with another: as the number of vectors increases in a fixed space of dimension $n$, there will be pairs with larger and larger correlations. However, our simulations demonstrate that (2) is much more important than (3). The tradeoff is dependent on the constants $\sigsq_z$, the importance of the unobserved covariates to the response (i.e. conceptually the same as $R^2$ for the observed covariates) and $\beta$, the size of the experimental effect.

In our previous work \citep{Kapelner2020} we studied the mean squared error (MSE) of the same estimator in the same settings in an effort to understand how it is affected by the design $D$. We found (a) the expected MSE over $\z$ is optimized when using one allocation vector that minimized $\abss{B_x}$, (b) when considering the worst case MSE by $\z$, BCRD is the optimal strategy and (c) when examining a high quantile of MSE over $\z$, then a design that provides good imbalance while having orthogonal designs is preferred. The results about experimental power dovetail with our previous finding in (c).

%\subsection{Recommendations for Practice}

Thus, based on this work and our previous work, we offer design recommendations for the practitioner that address these salient considerations. 

We recommend using rerandomization with a threshold as low as possible to produce $R$ on the order of a few thousand. The threshold can be lowered by having more time and computational resources at the practitioner's disposal. Rerandomization puts an upper limit on observed covariate imbalance and also provides allocations which seem (based on our numerical experience) to be as orthogonal as allocations that are expected in BCRD (we await a rigorous proof of this statement). 

We also recommend greedy pair switching \citep{Krieger2019} which provides smaller observed imbalance than rerandomization (by an order of $n$) and has orthogonality of its allocations proven to be nearly as orthogonal as BCRD in finite sample sizes.

A priori pairwise matching has observed covariate imbalance performance between these two designs but is not recommended because the orthogonality of its allocations is higher, lowering power in the small sample setting (see Figure~\ref{fig:average_power_sims}, columns 1 and 2). However, if the linear additive model (Equation~\ref{eq:simple_model}) is not believed, matching affords better performance on consideration (1) as it will match the response component of the observed covariate and not just the value of observed covariate \citep[see][Section 2.3.2]{Kallus2018}. This is recommended as lower allocation orthogonality does not penalize power by a significant amount.

Designs that enumerate all allocations vectors in $\allocspace_D$ and select $\w$ that provide optimally small $\abss{B_x}$ are possible in very small sample sizes (up to $n \approx 30$) and naively would seem to be the best strategy. We recommend strongly against these designs as they both have small $R$ and highly dependent allocations which can be deleterious to the experiment (see Figure~\ref{fig:average_power_sims}, column 1, green line).

\subsection{Further Research}\label{subsec:extensions}

%%%%

In this paper, we considered one covariate but in more realistic settings, there would be many covariates. In such settings, this would increase the imbalance in the observed covariates and increase the importance of term (I) in Equation~\ref{eq:betahat_diff} but would not impact the other terms which are functions of the dependence structure of the allocation vectors. Our design recommendation remains the same: find allocations that drive down covariate imbalance but also retain the relative independence of the allocations. To reiterate, this can be accomplished by using a lower rerandomization threshold. Future work can investigate numerical optimization approaches that optimize observed imbalance \citep{Bertsimas2015, Kallus2018} while still providing many allocations with a high degree of orthogonality.

Also, we assumed the classic differences-in-mean estimator $\betaThat = \half(\YbarT - \YbarC)$. Alternatively, one could employ the OLS estimator which adjusts for the observed covariate(s). We showed in our previous work that the MSE of the estimator is an entire order of $n$ lower in observed covariate imbalance \citep[Equation 24]{Kapelner2020}. Thus, when the OLS estimator is employed, we anticipate term (I) to decrease but terms (II) and (III) to be unaffected. This will result in the number of vectors $R$ and the average absolute allocation dependence $\rho$ having a more pronounced impact on experimental power.

%Our metric studied is $\POW$, the expectation of power over $B_z$, the asymptotic distribution of the imbalance in the unobserved response component, and the subset of vectors employed $\allocspaceR$. \citet{Kapelner2020} considered the impact of the worst-case $\z$ on MSE as well as quantiles where the optimal design was BCRD and a design between BCRD and numerical optimization respectively. The latter trades off $\abss{B_x}$ versus $\rho$. Here, the worst-case $B_z$ is not interesting as there is always a realization of $B_z$ that will result in zero power. But examing a low quantile over $B_z$ is interesting. We believe that using a low quantile of power as our metric will not result in a similar tradeoff between randomization and numerical optimization. The difficulty in computing the quantile is we have to count the number of allocations that will be rejected (expectation being linear makes our metric easier to analyze). Since each indicator doesn't behave with this tradeoff. \inred{David please edit}

We also assumed a basic response model, Equation~\ref{eq:simple_model}, which is additive and linear in the observed covariate, unobserved covariate and treatment effect. If this model is incorrectly specified, this would be the same as the $\z$ term being a function of $\x$ and thus $\x$ would appear in the expression $B_z$ and its finite-sample central limit theorem. Examining designs in this setting would be interesting future work.

In the approximation scenario where $2R \ll R_D$ and a Monte Carlo, recommending $R$ to be large has been noted in the literature on permutation tests \citep[Section 3.2]{Hemerik2018}. An additional concern in the approximation scenario, our power expression given in Equation~\ref{eq:pow} implicitly assumes a $p$-value calculation of the ratio of number of elements beating the run estimate to the total number of allocations, known in the literature as the \qu{unbiased estimate}. \citet{Phipson2010} and many others caution practitioners of this approach as it is anti-conservative by a factor of about $1 / (2R)$ for $\alpha < 1/2$, an unintuitive flaw that compounds in severity during multiple testing. Their solution is to use the Wilson estimate, adding one to both the numerator and denominator in the ratio calculation, which is conservative if $R$ is small. However, since in our construction we include $m_{ii}$ in the quantile calculation, our test then becomes conservative \citep[Equation 15.8]{Lehmann2006}. However, since randomization tests are categorically different methods that do not rely on group structure \citep{Hemerik2019}, we are unsure how this literature applies to our findings. We leave exploration of these issues to further research although we do not believe a new estimator or dropping the identity permutation will change our results nor the thrust of our recommendations.

There is a rich literature in science where the applications require finding allocation vectors with small dependence similar to the problem we face in our application. For example, if a Hadamard matrix of order $n$ exists, then its rows yield $n$ orthogonal allocations for an experiment of size $n$. However, our focus is small experiments (e.g. $n < 100$) and thus this number of allocations falls far short of the desired $R > 1000$ (even before filtering for covariate balance). Other methods such as the \citet{Gold1967} codes and \citet{Kasami1966} codes that are used in telecommunication and GPS technology offer allocation sets which are larger, but not by the orders of magnitude we have shown are necessary. Broadly, when $R$ is large relative to $n$, attempts to eliminate correlation are fruitless in the following sense: For allocation vectors ${w_1,\ldots, w_R}$ in an experiment of size $n$, \citet[Theorem 2.2]{Datta2012} proves that

\beqn
\frac{1}{R(R-1)} \sum \limits_{i \neq j} |w_i^Tw_j|^2 \,\geq\, n\frac{R - n}{R-1}.
\eeqn

\noindent This implies that $R \gg n$ this general lower bound on the correlations approaches $n$. Noting that $n$ is achieved in expectation by the complete randomization design, we believe very little improvement is possible in the standard experiment model considered in this paper. But, practical experiments may have constraints on the set of feasible allocations which induce correlations in naive draws and the allocations chosen by the rerandomization methods discussed herein. In that case, optimizing designs jointly for orthogonality of allocations and covariate balance can yield gains, based on the results in Section (reference where we talk about rho). We leave this to future work.

%We can also think about optimizing designs by using a subset of R picked so that $|r_{ij}|$'s are smaller. However, this is at the expense of $R$ being smaller. That seldom will matter and it will be computationally expensive.

Incidence and survival endpoints are also important and left for further research. Our intuition is the same considerations will be present and our design recommendation would be similar.

\bibliographystyle{apacite}
\bibliography{refs}

\pagebreak

\appendix

\section{Technical Results and Proofs for \qu{\ourtitle}}\label{app}
\renewcommand{\theequation}{A\arabic{equation}}
\renewcommand{\thefigure}{A\arabic{figure}}

The following is a table of notation we use throughout the paper (above the horizontal line) and in the appendix (below the horizontal line).

\begin{table}[h]
\centering
\begin{tabular}{cl}
\hline
Expression & Description \\ \hline
$n$ & The sample size \\
$\alpha$ & The Type I error of the randomization test assumed $< 1/2$ \\
$\beta$ & The additive treatment effect in the response \\
$\w, -\w$ & An allocation vector of length $n$ and its mirror \\
$r_{ij} := \w_i^\top \w_j / n$ & The correlation of two allocations \\
$B_{x,i} := \w_i^\top \x / n$ & The imbalance in an observed covariate \\
$~B_{z,i} := \w_i^\top \z / n$ & The imbalance in an unobserved response component \\
$\allocspace_D$ & The set of allocation vectors in the strategy $D$. \\
$\allocspaceR$ & A subset of $R$ unmirrored vectors from $\allocspace_D$ with their $R$ mirrors \\
$R$ & The number of unmirrored allocation vectors in the strategy $\allocspaceR$ \\
$\phi(t;\mu,\sigsq)$ & The PDF of the normal distribution with mean $\mu$ and variance $\sigsq$\\
$\phi(t) := \phi(t;0,1)$ & The PDF of the standard normal distribution \\
$\Phi(t)$ & The CDF of the standard normal distribution \\
$F_B(x, n, p)$ & The CDF of the binomial distribution \\
$\rho$ & Average absolute correlation between the allocations \\ \hline
$\gamma = \sqrt{n} \beta$ & The sample size-scaled treatment effect \\
$Z_i$ & $\iid$ standard normal random variables where $i = 0,1,2,\ldots$ \\
$~B_{Z_i} := \sqrt{n} B_{z,i}$ & The standardized imbalance in an unobserved response component \\
%$B_{Z_i} = \sqrt{\rho}Z_0 + \sqrt{1 - \rho} Z_i$ & A convenient expression for the imbalance in the \\ & unobserved response component \\
%$q(z)$ & The quantile such that $\Phi\parens{\frac{\sqrt{\rho}z+\rho \gamma-q(z)}{{\sqrt{1-\rho}}}}+\Phi\parens{-\frac{\sqrt{\rho}z+\rho \gamma+q(z)}{{\sqrt{1-\rho}}}} = 2\alpha$ \\
$q = \floor{2\alpha R}-1$ & The number of vectors to beat to reject the null hypothesis \\
\hline
\end{tabular}
\end{table}

\begin{theorem}\label{thm:toy_case}
For $R=2$, $\abss{r_{12}} = \rho$ and $\alpha = 1/4$, $\POW$ is monotonically decreasing in $\rho$.
\end{theorem}

\begin{proof}
By Assumption~\ref{ass:Bznorm} and Equation~\ref{eq:corrBziBzj}, the imbalances of the unobserved covariate among the two designs has a bivariate normal distribution,

\bneqn\label{eq:imbalance_z_mvnorm}
\twovec{B_{Z_1}}{B_{Z_2}} \sim \multnormnot{2}{\zerovec}{\twobytwomat{1}{\rho}{\rho}{1}}.
\eneqn

\noindent The estimate for the treatment effect for the design that generated the data can be written as

\beqn
B_{Z_1}=Z_1+\gamma,
\eeqn

\noindent The estimate based on its mirror image is then $-B_{Z_1}$. Similary the estimate for the other choice of assignment is

\beqn
B_{Z_2}=Z_2+\rho \gamma,
\eeqn

\noindent and its mirror image produces an estimate of $-B_{Z_2}$ for the treatment effect.

It is then algebra to find the set of $Z_1,Z_2$  for which simultaneously,  $B_{Z_1} > -B_{Z_1}$, $B_{Z_1}>B_{Z_2}$ and $B_{Z_1}>-B_{Z_2}$. This results in $Z_1>-\gamma$ and

\beqn
-Z_1-(1+\rho)\gamma<Z_2<Z_1+(1-\rho)\gamma.
\eeqn

\noindent Hence the objective we want to consider as a function of $\rho$, $G(\rho)$, which denotes the probability that $B_{Z_1}$ provides the largest estimate among the four choices, is

\beqn
G(\rho)=\int_{-\gamma}^\infty \parens{ \Phi\Big(\frac{z_1+(1-\rho)\gamma - \rho z_1}{\sqrt{1-\rho^2}}\Big)-\Phi\Big(\frac{-z_1-(1+\rho)\gamma-\rho z_1}{\sqrt{1-\rho^2}}\Big) } \phi(z_1)dz_1,
\eeqn

\noindent as $Z_2\,|\,Z_1$ has a normal distribution with mean $\rho z_1$ and varinace $1-\rho^2$. If we reparameterize $\rho$ to $\omega=\sqrt{\frac{1+\rho}{1-\rho}}$, and $u=z_1+\gamma$, 

\[
G(\omega)=\int_{0}^\infty \Big( \Phi(u/\omega)-\Phi(-u\omega)\Big)\phi(u-\gamma)du = \int_{0}^\infty \Big( \Phi(u/\omega)+\Phi(u\omega)\Big)\phi(u-\gamma)du -\Phi(\gamma) .
\]

Now $\omega$ indexes the agreement between the two designs. If $\rho=-1$, then the designs are mirror image and $\omega=0$. If $\rho=0$, which corresponds to half of the treated in one design are assigned to treatment in the other design and similarly for the control, likely to be the case that maximizes $G(\omega)$, then $\omega=1$ by Corollary~\ref{cor:g_omega_max}. Finally, as the two designs become more similar, $\rho$ is positive and goes to one and $\omega$ corresponding to $\omega>1$ and increasing to infinity. 

One can take the derivative and explore $G(\omega)$'s behavior. It follows that the derivative of $G(\omega)$, $g(\omega)$, is

\begin{equation}
\label{derg}
g(\omega)=\int_{0}^\infty u \parens{\phi(u \omega)-\frac{\phi(u/\omega)}{\omega ^2}} \phi(u-\gamma)du.
\end{equation}

In order to evaluate Equation~\ref{derg}, consider

\begin{equation}
\label{gh}
h(a)=a\int_{0}^\infty u\phi(au)\phi(u-\gamma)du.
\end{equation}

\noindent which yields the desired result using Lemma~\ref{lemm:toy_case}.

\end{proof}

\begin{lemma}\label{lemm:toy_case}

\bneqn\label{ghl}
\int_{0}^\infty u\phi(au)\phi(u-\gamma)du = \frac{e^{-\overtwo{\gamma^2}}}{2\pi(1+a^2)}+\frac{\gamma e^{-\frac{1}{2}\frac{a^2\gamma^2}{1+a^2}}}{\sqrt{2\pi} (1+a^2)^{3/2}}\parens{1-\Phi\parens{-\frac{\gamma}{\sqrt{1+a^2}}}}.
\eneqn

\end{lemma}

\begin{proof}

Since

\beqn
\phi(au)\phi(u-\gamma)=\frac{1}{2\pi}e^{-\frac{1}{2}a^2u^2}e^{-\frac{1}{2}(u-\gamma)^2},
\eeqn

\noindent by completing the square

\beqn
\phi(au)\phi(u-\gamma)=\frac{1}{2\pi}e^{-\frac{1}{2}(1+a^2) \squared{u-\frac{\gamma}{1+a^2}} }e^{-\frac{1}{2}\frac{a^2\gamma^2}{1+a^2}}.
\eeqn

\noindent Substituting this expression into Equation~\ref{gh} yields 

\bneqn\label{ghi}
h(a)=\frac{a}{2\pi}e^{-\frac{1}{2}\frac{a^2\gamma^2}{1+a^2}}\int_{0}^\infty ue^{-\frac{1}{2}(1+a^2) \squared{u-\frac{\gamma}{1+a^2}} }du.
\eneqn

\noindent By letting $v=\sqrt{1+a^2}(u-\frac{\gamma}{1+a^2})$, the integral in Equation~\ref{ghi} becomes,

\bneqn\label{ghii}
\frac{1}{\sqrt{1+a^2}}
\int_{-\frac{\gamma}{\sqrt{1+a^2}}}^\infty \frac{v}{\sqrt{1+a^2}}e^{-\frac{1}{2}v^2}dv
+
\frac{\gamma}{(1+a^2)^{3/2}}
\int_{-\frac{\gamma}{\sqrt{1+a^2}}}^\infty e^{-\frac{1}{2}v^2}dv.
\eneqn

\noindent Then, the first integral in Equation~\ref{ghii} can be shown to be 

\beqn
\frac{1}{1+a^2}e^{-\frac{1}{2}\frac{\gamma^2}{1+a^2}}
\eeqn

\noindent  and the second integral can be shown to be 

\beqn
\sqrt{2\pi}\parens{1 - \Phi\parens{-\frac{\gamma}{\sqrt{1+a^2}}}}.
\eeqn

\noindent Substituting these two results into Equation~\ref{ghi} completes the proof.

\end{proof}

\begin{corollary}\label{cor:g_omega_max}
$G(\omega)$ is maximized when $\omega=1$ for all $\gamma > 0$.
\end{corollary}

\begin{proof}

Since $G(\omega)=G(1 / \omega)$, it is sufficient to show that $g(\omega)$ is negative for all $\omega>1$. But 

\[
g(\omega)=\oneover{\omega}\parens{h(\omega)-h\parens{1 / \omega}}.
\]

Since $\frac{a}{1+a^2}$ is equal for $\omega$ and $1/\omega$ so the first term in Equation~\ref{ghl} cancels, $\frac{a\gamma}{(1+a^2)^{3/2}}$ decreases in $a$, and $-\frac{a^2}{1+a^2}$ decreases in $a$,  and $1-\Phi(-\frac{\gamma}{\sqrt{1+a^2}})$ also decreases in $a$, the result follows. 
\end{proof}

\begin{theorem}\label{thm:power_asymptotic_in_R}
The asymptotic power is

\bneqn\label{eq:asymptotic_power_integral}
\lim_{R \rightarrow \infty} \POW = \int_\reals \Phi\parens{\frac{\gamma + \sqrt{\rho} z - q(z)}{\sqrt{1 - \rho}}} \phi(z)dz
%\int_\reals  \Phi(\gamma-q(z))\phi(z)dz
\eneqn

\noindent where $q(z)$ is the fixed quantity satisfying

\bneqn\label{eq:qz}
\Phi\parens{\frac{\sqrt{\rho}z+\rho \gamma-q(z)}{{\sqrt{1-\rho}}}}+ 
\Phi\parens{-\frac{\sqrt{\rho}z+\rho \gamma+q(z)}{{\sqrt{1-\rho}}}} = 2\alpha.
\eneqn
\end{theorem}

\begin{proof}

%The question is what happens more generally with $R$ of reasonable size. It is too difficult to explore the problem in full generality, allowing for different $r_{ij}$ as there are too many of them. Attention will be restricted to the case of $r_{ij}=\rho$ for $\binom{R}{2}$ pairs.

%There are $R$ assignment vectors along with their $R$ mirror images. This implies that $B_z$ are multivariate normal as in Equation~\ref{eq:imbalance_z_mvnorm}. The power goes to a limit as $R \rightarrow \infty$ given by the following theorem.

We first find the asymptotic $\quantile{\m_{i\cdot}}{1 - \alpha}$ by row. Assume that the allocation that generates $\y$ is either $\w_1$ or $-\w_1$. The treatment estimate under the allocation $\w_1$ we define as

\beqn
V_{11} = B_{Z_1} + \gamma
\eeqn

\noindent and $V_{11}^m = -V_{11}$, the estimate under its mirror allocation $-\w_1$. We need to find the $1-\alpha$ quantile for the treatment estimates of the other $2(R-1)$ designs. We define

\bneqn\label{eq:defV}
V_{1j} = B_{Z_j} +\rho \gamma
\eneqn

\noindent where $j = 2,\ldots,R$ as the elements of $\m_1$ without the contribution of the observed covariate. We define $V^m_{1j} = -V_{1j}$, for the mirror allocation. %We also let $I_{1j}$ be the indicator that $V_{1j}$ beats the diagonal element $V_{11}$ and $I_{1j}^m$ be the indicator that its mirror beats the diagonal element $V_{11}$.

A convenient expression for the imbalance in the unobserved response component is

\bneqn\label{eq:decompBz}
B_{Z_i} = \sqrt{\rho}Z_0 + \sqrt{1 - \rho} Z_i.
\eneqn

\noindent If we condition on $Z_0=z$, the $V_{1j}$ are iid. Finding the asymptotic quantity $q(z)$ requires solving the implicit equation $\prob{V_{1j} > q(z)} + \prob{-V_{1j} > q(z)}=2\alpha$. 

The expected power is identical to the power for any row. Thus,

\bneqn\label{eq:lim_pow_z_W2R}
\lim_{R \rightarrow \infty} \POW_{\z, \allocspaceR} =  \expe{I_{1}}
\eneqn

\noindent where $I_1$ is the indicator that is one if the diagonal $V_{11}$ exceeds $q(z)$ conditional on $Z_0 = z$. This expectation corresponds to the normal CDF expression in the integrand of the statement of the theorem, Equation~\ref{eq:asymptotic_power_integral}. Integrating over all $z$ from $Z_0$ completes the proof.

\end{proof}

%\begin{lemma}\label{lemm:power_asymptotic_in_rho}
%The asymptotic power decreases in $\rho$ if $\alpha<\frac{1}{2}$.
%\end{lemma}
%
%\begin{proof}
%Consider taking the derivative of $q(z)$ with respect to $\rho$. Since $\alpha<\frac{1}{2}$, clearly $q(z) \ge 0$ and $\half\parens{\Phi(q(z)+\gamma\rho)+\Phi(q(z)-\gamma \rho)}$ increases in $q(z)$ as $\phi(q(z)+\gamma\rho)<\phi(q(z)-\gamma\rho)$. This implies that $q(z)$ increases if  $\rho$ increases  and hence $\Phi(\gamma-q(z))$ decreases if $\rho$ increases.\\
%\end{proof}

\begin{theorem}\label{thm:power_fixed_R_expression}

Power can be computed via 

\bneqn\label{eq:pow_integral}
\POW = \int_0^\infty \int_\reals F_{B}(q; R-1, p(u, s)) \phi(s; 0 ,\rho) \phi\parens{u; \frac{\gamma}{\sqrt{1 - \rho}}, \frac{1}{1 - \rho}} dsdu
\eneqn

\noindent where $q = \floor{2\alpha R}-1$ and

\bneqn\label{eq:pvu}
p(u, s) = \Phi\parens{-\parens{1-\rho} u + s} + \Phi\parens{-\parens{1+\rho}u - s}.
\eneqn

%\noindent $F_B(\cdot, \cdot, \cdot)$ is the CDF of the binomial random variable and $\phi(\cdot)$ and $\Phi(\cdot)$ are the PDF and CDF of the standard normal random variable respectively.

%%%%%%%%%%%%

\end{theorem}

\begin{proof}

In the definition of power conditional on $\z$ and an arbitrary subset of $\allocspaceR \subset \allocspace_D$ (Equation~\ref{eq:pow}) we average the indicator that the row's diagonal element beats the appropriate row quantile, i.e.

\beqn
\POW_{\z, \allocspaceR} := \oneover{2R} \sum_{i=1}^{2R} \indic{m_{i,i} > \quantile{\m_{i\cdot}}{1 - \alpha}}.
\eeqn

\noindent Assuming that $\rho = \abss{r_{ij}}$ makes all subsets of $\allocspace_D$ the same in the context of this calculation and all rows the same in expectation. Hence, we can consider just the first row. Letting $V_{1,j}$  and $V_{1,j}^m$ denote the treatment estimates for allocations and their mirrors (as in Equation~\ref{eq:defV}), we have

\bneqn\label{eq:V11quantile}
\POW = \expesub{B_z}{\prob{V_{1,1} > \quantile{\braces{V_{1,1}, V_{1,1}^m, V_{1,2}, V_{1,2}^m, \ldots, V_{1,R}, V_{1,R}^m}}{1 - \alpha}}}.
\eneqn

\noindent To understand this quantile, we must understand the probabilities that $V_{1,1}^m$, $V_{1,j}$, or $V_{1,j}^m$ exceeds $V_{1,1}$. 

To simplify these calculations, we now condition on $Z_0 = z_0$ and $Z_1 = z_1$ (since $B_{z_j}$ can be written in terms of $Z_0$ and $Z_1$ as in Equation~\ref{eq:decompBz}). This allows all bivariate pairs $(V_{1,1}, V_{1,j})$ to be iid. If $V_{1,1} < 0$, then necessarily either $V_{1,j} > V_{1,1}$ or $V_{1,j}^m > V_{1,1}$ for all $j = 2, \ldots, R$ then $V_{1,1}$ can never be higher than the quantile of interest in Equation~\ref{eq:V11quantile} (since $\alpha < 1/2$ by assumption). 

Since $V_{1,1} > 0$, we can ignore its mirror as $V_{1,1} > V_{1,1}^m$. For every pair $(V_{1,1}, V_{1,j})$ we compute the probability that $V_{1,j} > V_{1,1}$ or $V_{1,j}^m > V_{1,1}$. Viewing each of these $R - 1$ events as iid Bernoulli random variables, we have a binomial random variable counting the total number of such events, reflecting the number of allocations that beat the experimental run. We want this number of events to be sufficiently small (less than the fraction $\alpha$ of the total number of allocations) in order for $V_{1,1}$ to beat the quantile and reject $H_0$.

%Clearly $V_{1,1} > V_{1,1}^m$ and $V_{1,j}^m > V_{1,1}$ iff $V_{1,1}<0$. If $V_{1,1}>0$ then $V_{1,1}^m<V_{1,1}$. 
We now compute the probability parameter in this Bernoulli event. The case when $V_{1,j}>V_{1,1}$ occurs if

\beqn
Z_j > z_1 + \gamma\sqrt{1-\rho} 
\eeqn 

\noindent and the case when $V_{1,j}^m > V_{1,1}$ occurs if

\beqn
Z_j < -z_1-\frac{2\sqrt{\rho}z_0+(1+\rho) \gamma}{\sqrt{1-\rho}}.
\eeqn 

\noindent In summary, $\prob{Z_0, Z_1} = 1$ if $V_{1,1} < 0$ and

\bneqn\label{eq:pZ0Z1}
\prob{Z_0, Z_1} = \Phi\parens{-Z_1 - \sqrt{1 - \rho} \gamma} + \Phi\parens{-Z_1 - \frac{2\sqrt{\rho}Z_0+(1+\rho)\gamma}{\sqrt{1-\rho}}}.
\eneqn

\noindent if $ V_{1,1} > 0$. We now make a change of variables from $(Z_0, Z_1)$ to $(U, S)$ where 

\bneqn\label{eq:u_and_s_change_of_variables}
U &=& Z_1+\frac{\sqrt{\rho}Z_0+\gamma}{\sqrt{1-\rho}}, \\
S &=& -\rho Z_1+\sqrt{\rho(1-\rho)}Z_0.
\eneqn

\noindent In this change of variables, $U$ and $S$ are independent and if $ V_{1,1}>0$ then $U>0$. Also

\beqn
\expe{U} &=& {\gamma}/ \sqrt{1-\rho}, \\
\expe{S} &=& 0, \\
\var{U} &=& 1 / (1 - \rho), \\
\var{S} &=& \rho
\eeqn

\noindent and

\beqn
-(1 - \rho)U + S &=&  -Z_1 -\sqrt{1 - \rho}\gamma, \\
-(1 + \rho)U - S &=& -Z_1 - \frac{2\sqrt{\rho}Z_0 + (1 + \rho) \gamma}{\sqrt{1-\rho}}.
\eeqn

\noindent Substituting the above into Equation~\ref{eq:pZ0Z1} gives us the definition of Equation~\ref{eq:pvu} and integrating over the normal densities for $(U,S)$ in the appropriate region completes the proof.

\end{proof}

\begin{lemma}\label{lemm:uniform_for_general_densities}

Consider $V_i := V(X_0, X_i) := X_0 + X_i$ where  $X_0 \sim g$ and $X_i \iid f$ (where $i = 1, 2$)  are continuous random variables symmetric about zero and $X_0$ is independent of $X_i$. Then the random variable $T$ defined as the probability

%with densities $g(x)$ and $f(x)$ (for all $i = 1, 2, \ldots$) respectively and all random variables are mutually independent. Then the random variable $Y$ defined as the probability

\beqn
T(X_0, X_1) := \cprob{|V_2|>V_1}{X_0,X_1}
\eeqn 

\noindent has point mass of $1/2$ at $1$ and is otherwise uniform. 
\end{lemma}

\begin{proof}

Note that $V_i$ are random variables that are symmetric about 0. We show the result conditional on $X_0=Ia$ where $I$ is $\pm 1$ with probability $1/2$ and $a \neq 0$ is a real constant. Furthermore, the
behavior of $|V_i|$ over $X_i$ is the same for $X_0=a$ and $X_0=-a$. Therefore, let 

\beqn
\bar{F_a}(t)=\prob{|V(a,X)| > t} = \prob{|V(-a,X)| > t}
\eeqn

\noindent and thus

\beqn 
T(Ia, X_1)=\bar{F_a}(V(Ia, X_1))
\eeqn

\noindent since $\bar{F}(t) = 1$ when $t<0$ and  $\prob{V(Ia, X_1 < 0} = 1 / 2$ over $I$ and $X_1$. This demonstrates $T=1$ with probability $1/2$. Otherwise, we can replace $V(Ia, X_1)$ with its absolute value. So for $t<1$,

\beqn
\probsub{X_1,I}{T(Ia,X_1) \le t} &=& \frac{1}{2}\probsub{X_1,I}{\bar{F_a}(|V(a,X_1)|)\le t} \\
&=& \half \probsub{X_1,I}{|V(a,X_1)| \le \bar{F_a}^{-1}(t)} \\
&=&\frac{1}{2}\parens{1-\probsub{X_1,I}{\abss{V(a,X_1)} \ge \bar{F_a}^{-1}(t)}} \\
&=& \frac{1}{2}(1-(1-t)) = \frac{t}{2}.
\eeqn

\noindent The next to last equality follows since $\probsub{X_1,I}{|V((a,X_1)|>t} = \bar{F_a}(t)$. Since the above is true for every value of $X_0$, when we integrate over the symmetric density $g$, we obtain the desired result.

\end{proof}

\begin{corollary}\label{corr:size}
For all $\alpha \in (0,0.5)$, $R \in \braces{\floor{2\alpha n} \in \naturals\,:\, n \in \naturals}$ and $\rho \in (0,1)$,

\bneqn\label{eq:size}
\alpha = \int_{v>0}\int_{u \in \reals} F_{B}(\floor{2\alpha R} - 1; R-1, p(v, u)) \phi(u)\phi(v) dudv
\eneqn

\noindent where $p(v, u)$ is defined as in Equation~\ref{eq:pvu}.
\end{corollary}

\begin{proof}
This is a special case of Lemma~\ref{lemm:uniform_for_general_densities} where $X_0 \sim \normnot{0}{\rho}$ and $X \sim \normnot{0}{1 - \rho}$. This result is intuitive as Equation~\ref{eq:pow_integral} is the power of the randomization test with no experimental effect (i.e. $\gamma$ = 0) thus equal to the size of the test $\alpha$.
\end{proof}

\begin{lemma}\label{lemm:main_support}
Consider $\rho \in \zeroonecl$, $V_1 := \sqrt{1-\rho}Z_1+\sqrt{\rho}Z_0+\gamma$ and $V_2 := \sqrt{1-\rho}Z_1+\sqrt{\rho}Z_0+\rho \gamma$ where  $Z_0, Z_1, Z_2 \iid \stdnormnot$. Then the random variable $T$ defined as the probability

\bneqn\label{eq:def_rv_T}
T(Z_0, Z_1) := \cprob{|V_2|>V_1}{Z_0, Z_1}
\eneqn 

\noindent has a point mass of $\Phi(-\gamma)$ at $t=1$ and otherwise has density

%with densities $g(x)$ and $f(x)$ (for all $i = 1, 2, \ldots$) respectively and all random variables are mutually independent. Then the random variable $Y$ defined as the probability

\bneqn\label{eq:pdfT}
f_T(t; \rho, \gamma) =e^{-\frac{1}{2}\gamma^2} \int_0^\infty e^{\gamma\sqrt{1-\rho} u(a,t) }  \phi\parens{a; \, 0, \textstyle\frac{\rho}{1-\rho}} da
\eneqn

\noindent where $u = u(a, t)$ satisfies 

\bneqn\label{eq:uat}
\Phi(-u + a) + \Phi(-u - a) = t.
\eneqn

\end{lemma}

\begin{proof}
The point mass can be explained as follows: if $V_1 \leq 0$, then $|V_2|$ is certainly greater than $V_1$ i.e. $\prob{|V_2|>V_1} = 1$ which happens with $\prob{V_1 \leq 0} = \Phi(-\gamma)$.

We now focus on the density $f_T$ which involves understanding the probability $T(Z_0, Z_1)$. We first make the same change of variables from $(Z_0, Z_1)$ to $(U, S)$ as in in Theorem~\ref{thm:power_fixed_R_expression},

\beqn
U &=& Z_1+\frac{\sqrt{\rho}Z_0+\gamma}{\sqrt{1-\rho}}, \\
S &=& -\rho Z_1+\sqrt{\rho(1-\rho)}Z_0.
\eeqn

\noindent With this transformation, $T$ is now the random variable form of Equation~\ref{eq:pZ0Z1},

\beqn
T(U, S) = \Phi\parens{-(1-\rho)U + S} + \Phi\parens{-(1 + \rho)U - S}.
\eeqn

\noindent To explore the behavior of $T(U, S)$, we first examine the values of $S$ for fixed $U>0$ for which $T(U, S) \le t$. But for fixed value of $U>0$, it follows that $T(U, S)$ increases iff 

\beqn
-(1-\rho)U+S>-(1+\rho)U-S \mathor S>-U\rho. 
\eeqn

\noindent We want $T(U, S) \le t$. But $T(U, S)$ is minimized for fixed $U>0$ when $S=-U\rho$. Since $\prob{U, -U\rho} = 2\Phi(-U)$, It follows that for given $t$  the values $U$  we need to consider are  $U \ge \Phi^{-1}(t/2)$. 

Since we showed that for fixed $U$,  $T(U, S)$ is minimized when $S = -U\rho$ and it increases as $S$ moves away from this value on both sides, therefore the set of $S$ for which $T(U, S) \le t$ is of the form: $-U\rho -b < S < -U\rho+a$ where $T(U, -U\rho + a) = T(U, -U\rho - b) = t$. Since $T(U, -U\rho + a) = \Phi(-U+a)+\Phi(-U-a)$, this implies that $b=-a$ and $a>0$ satisfies $\Phi(-U+a)+\Phi(-U-a)=t$, the 1:1 implicit relationship given in Equation~\ref{eq:uat}.

It is important to note that $a$ does not depend on $\rho$. Since 

\beqn
S \sim \normnot{\textstyle \frac{\gamma}{\sqrt{1-\rho}}}{\textstyle \oneover{1-\rho}}
\eeqn

\noindent then we can compute the CDF of $T$ by integrating over the conditional CDF of $T\,|\,U$,

\beqn
F_T(t; \rho, \gamma) := \prob{T(U, S) \le t} = \sqrt{1-\rho} \int_{-\Phi^{-1}(t/2)}^\infty F_{T\,|\,U}(t, u; \rho) \phi(\sqrt{1-\rho}u - \gamma)du.
\eeqn

\noindent We simplify the $\phi$ expression using algebra to arrive at

\bneqn\label{eq:cdfT}
F_T(t; \rho, \gamma) = e^{-\frac{1}{2}\gamma^2} \sqrt{1-\rho} \int_{-\Phi^{-1}(t/2)}^\infty F_{T\,|\,U}(t, u; \rho)  e^{\gamma \sqrt{1 - \rho} u} \phi(\sqrt{1-\rho} u) du.
\eneqn

\noindent and the conditional CDF can be computed via

\bneqn\label{eq:cdfTcondU}
F_{T\,|\,U}(t, u; \rho, \gamma) = \prob{-u\rho - a < S < -u\rho + a} = \Phi\parens{\frac{-u\rho+a}{\sqrt{\rho}}} - \Phi\parens{\frac{-u\rho-a}{\sqrt{\rho}}}
\eneqn

\noindent where $a$ satisfies Equation~\ref{eq:uat}. To obtain the density, we differentiate Equation~\ref{eq:cdfT} with respect to $t$. By Leibnitz's formula,

\beqn
f_T(t; \rho, \gamma) &=& e^{-\frac{1}{2}\gamma^2} \sqrt{1-\rho}~ \bigg(\int_{-\Phi^{-1}(t/2)}^\infty \frac{d}{dt}\bracks{F_{T\,|\,U}(t, u; \rho)  e^{\gamma \sqrt{1 - \rho} u} \phi(\sqrt{1-\rho} u)} du + \\
&&  ~~~~~~~~~~~~~~~~~~~~\Phi^{-1}(t/2) I\parens{\Phi^{-1}(t/2)} \frac{d}{dt}\bracks{\Phi^{-1}(t/2)}  \bigg).
\eeqn

\noindent where $I(\cdot)$ is the evaluated integrand.

We first note that if $u = -\Phi^{-1}(t / 2)$, then $a=0$ for all $t$ and hence $F_{T\,|\,U}(t, u; \rho, \gamma)=0$ for all $t$ and thus the evaluated integrand term $I(\cdot)$ is zero. This fact plus substituting Equation~\ref{eq:cdfTcondU} in the above gives us

\beqn
f_T(t; \rho, \gamma) = e^{-\frac{1}{2}\gamma^2} \sqrt{1-\rho} \int_{-\Phi^{-1}(t/2)}^\infty e^{\gamma \sqrt{1 - \rho} u} \phi(\sqrt{1-\rho} u) \underbrace{\frac{d}{dt}\bracks{\parens{\Phi\parens{\textstyle \frac{-u\rho+a}{\sqrt{\rho}}} - \Phi\parens{\textstyle \frac{-u\rho-a}{\sqrt{\rho}}}}}  }_D du.
\eeqn

\noindent We now evaluate the derivative term $D$. Since it is within the integral, the variable $u$ is fixed. We use

\bneqn\label{eq:dadt}
%\frac{d}{dt}\bracks{\parens{\Phi\parens{\textstyle \frac{-u\rho+a}{\sqrt{\rho}}} - \Phi\parens{\textstyle \frac{-u\rho-a}{\sqrt{\rho}}}}}
\frac{da}{dt} = \oneover{\phi(-u + a) - \phi(-u - a)}
\eneqn

\noindent to obtain

\beqn
D = \oneoversqrt{\rho} \parens{\phi\parens{\textstyle \frac{-u\rho+a}{\sqrt{\rho}}} + \phi\parens{\textstyle \frac{-u\rho-a}{\sqrt{\rho}}}} \frac{da}{dt}
= \oneoversqrt{\rho} \frac{
\phi\parens{\textstyle \frac{-u\rho+a}{\sqrt{\rho}}} + \phi\parens{\textstyle \frac{-u\rho-a}{\sqrt{\rho}}}
}{
\phi(-u + a) - \phi(-u - a)
}.
\eeqn

\noindent Combining the above with the $\phi$ expression in the integral we note the following algebraic simplification

\beqn
\oneoversqrt{\rho} \phi(\sqrt{1-\rho} u) \frac{
\phi\parens{\textstyle \frac{-u\rho+a}{\sqrt{\rho}}} + \phi\parens{\textstyle \frac{-u\rho-a}{\sqrt{\rho}}}
}{
\phi(-u + a) - \phi(-u - a)
} = \oneoversqrt{2\pi} \oneoversqrt{\rho} \frac{e^{2ua} + 1}{e^{2ua} - 1}  e^{- \frac{1}{2}\frac{1-\rho}{\rho} a^2}
\eeqn

\noindent Substituting this into the integral and simplifying we obtain

\beqn
f_T(t; \rho, \gamma) = e^{-\frac{1}{2}\gamma^2} \oneoversqrt{2\pi} \sqrt{\frac{1-\rho}{\rho}} \int_{-\Phi^{-1}(t/2)}^\infty e^{\gamma \sqrt{1 - \rho} u} \frac{e^{2ua} + 1}{e^{2ua} - 1}  e^{- \frac{1}{2}\frac{1-\rho}{\rho} a^2} du.
\eeqn

\noindent We now change variables from $u$ to $a$. 

\beqn
\frac{da}{du} = \frac{
\phi(-u + a) - \phi(-u - a)
}{
\phi(-u + a) + \phi(-u - a)
} = \frac{e^{2ua} + 1}{e^{2ua} - 1}
\eeqn

\noindent Substituting for $du$ and noting that the lower limit becomes zero, simplifies the density to

\beqn
f_T(t; \rho, \gamma) = e^{-\frac{1}{2}\gamma^2} 
\underbrace{\oneoversqrt{2\pi} \sqrt{\frac{1-\rho}{\rho}}}
\int_{0}^\infty e^{\gamma \sqrt{1 - \rho} u(a, t)} 
\underbrace{e^{- \frac{1}{2}\frac{1-\rho}{\rho} a^2} \vphantom{\sqrt{\frac{1-\rho}{\rho}}}}
da.
\eeqn

\noindent Noting that the underbraced terms above compose a normal density with mean 0 and variance $\rho / (1 - \rho)$ completes the proof.

\end{proof}

\begin{corollary}\label{corr:density_negative}
The random variable $T$ defined in Equation~\ref{eq:def_rv_T} of Lemma~\ref{lemm:main_support} has strictly decreasing density.
\end{corollary}

\begin{proof}
We take the derivative of the density (Equation~\ref{eq:pdfT}) as below:

\beqn
f'_T(t; \rho, \gamma) = e^{-\frac{1}{2}\gamma^2} \int_0^\infty  \phi\parens{a; \, 0, \textstyle\frac{1-\rho}{\rho}} \underbrace{\frac{d}{dt}\bracks{e^{\gamma\sqrt{1-\rho}u(a,t) }}}_D da.
\eeqn

\noindent All terms outside of $D$ are positive. Since $D = \gamma\sqrt{1-\rho}e^{\gamma\sqrt{1-\rho}u(a,t) } u'(a,t)$, to prove that $f'_T(t; \rho, \gamma) < 0$, it is sufficient to show that $u'(a,t) < 0$ since all other quantities composing $D$ are positive. Analogous to the calculation of Equation~\ref{eq:dadt} for fixed $a$ we have 

\beqn
\frac{du}{dt} &=& -\oneover{\phi(-u + a) + \phi(-u - a)} < 0 ~~\text{for all $a$ and $t$}.
\eeqn
\end{proof}

\begin{theorem}\label{thm:power_increases_monotonically_as_R_increases}
The power $\POW$ increases monotonically in $R$ as the attainable power increases (i.e. as the sequence of $R$ defined by $q = \floor{2\alpha R} \in \naturals$ increases where $\alpha < 1/2$ and $\alpha$ is valued appropriately to make the sequence possible).
\end{theorem}

\begin{proof}

Using the change of variables and the definition of $T$ found in Lemma~\ref{lemm:main_support} allows us to express power as

\beqn
%\bneqn\label{eq:pow_t}
\POW = \expesub{T}{F_{B}(q; n_q, T)} = \int_{0}^1F_B(q, n_q, t) f_T(t) dt.
%\eneqn
\eeqn

\noindent where $n_q$ computes the $q$ corresponding to $R - 1$. Even though this expectation is a Lebesgue integral over the measure of $T$, we can ignore the point mass at $t=1$ (because the power would be zero in that setting) resulting in the standard Riemann integral found above. 

To prove the theorem, it is sufficient to show that $\POW(q+1) > \POW(q)$ i.e.

\beqn
\POW(q+1) - \POW(q) = \int_{0}^1 (F_B(q + 1, n_{q + 1}, t) - F_B(q, n_q, t)) f_T(t) dt > 0.
\eeqn

\noindent Expressing the CDF of the binomial as a regularized incomplete beta function $I_x(\alpha, \beta)$ which is the ratio of an incomplete beta function $B(x; \alpha, \beta)$ to a beta function $B(\alpha, \beta)$,

\beqn
\POW(q+1) - \POW(q) &=& \int_{0}^1 (I_{1-t}(n_{q + 1} - q - 1, q + 2) - I_{1-t}(n_q - q, q + 1))  f_T(t) dt \\
&=& \int_{0}^1 \parens{\textstyle\frac{B(1-t; n_{q + 1} - q - 1, q + 2))}{B(n_{q + 1} - q - 1, q + 2))} -\frac{B(1-t; n_q - q, q + 1))}{B(n_q - q, q + 1))}}  f_T(t) dt \\
&=& \int_{0}^1 \underbrace{\int_{0}^{1 - t} \parens{\textstyle\frac{\theta^{n_{q + 1} - q - 2} (1-\theta)^{q + 1}}{B(n_{q + 1} - q - 1, q + 2))} -\frac{\theta^{n_q - q - 1} (1-\theta)^{q}}{B(n_q - q, q + 1))}} d\theta }_{h(t)} f_T(t) dt \\
\eeqn

\noindent We wish to explore the behavior of the inner integral $h(t)$. Taking the derivative with respect to $t$,

\beqn
h'(t) &=& \frac{(1 - t)^{n_q - q - 1} t^{q}}{B(n_q - q, q + 1))} - \frac{(1 - t)^{n_{q + 1} - q - 2} t^{q + 1}}{B(n_{q + 1} - q - 1, q + 2))}
\eeqn

\noindent The function $h(t)$ is increasing when the derivative is positive. After some algebra, this occurs when

\beqn
\frac{(1 - t)^{n_{q + 1} - q - 2} t^{q + 1}}{(1 - t)^{n_q - q - 1} t^{q}} < \frac{B(n_{q + 1} - q - 1, q + 2))}{B(n_q - q, q + 1))}
\eeqn

\noindent Simplifying the left side and noting that the right side is a positive constant $\eta$, the set of $t$ where $h(t)$ is increasing can be expressed simply as

\beqn
(1-t)^{n_{q + 1} - n_q - 1} t < \eta.
\eeqn

\noindent The exponent $n_{q + 1} - n_q - 1 = 1/(2\alpha) -1$ is a positive integer (since $\alpha < 1/2$ by assumption) and thus the left side above is a legal integrand in a beta function having the following three properties: (a) zero at $t = 0$ and $t = 1$ (b) always positive and (c) unimodal. This implies that $\eta$ intersects $h'(t)$ at two points which we denote $t_1$ and $t_3$ satisfying $0 < t_1 < t_3 < 1$. The set of $t$ for which $h(t)$ is increasing is then $[0, t_1] \cup [t_3, 1]$ and otherwise decreasing. Since $h(t)$ only has two critical points, it must be positive in $[0, t_2]$ and negative in $(t_2, 1]$ where $t_1 < t_2 < t_3$.

When $\gamma = 0$, $\POW(q+1) = \POW(q) = \alpha$ (Corollary~\ref{corr:size}) implying $\POW(q+1) - \POW(q) = 0$. Also if $\gamma = 0$, then $f_T(t)$ is constant (Lemma~\ref{lemm:uniform_for_general_densities}). These two facts together imply $\int_0^1 h(t) dt = 0$.  Since $f_T(t_2)$ is a lower bound for the density when $t < t_2$ and $f_T(t_2)$ is an upper bound for the density when $t > t_2$,

\beqn
\int_0^1 h(t) f_T(t) dt &=& \int_0^{t_2} f_T(t) h(t) dt + \int_{t_2}^1 f_T(t) h(t) dt \\
&\geq& f_T(t_2) \int_0^{t_2}  h(t) dt + f_T(t_2) \int_{t_2}^1  h(t) dt \\
&=& f_T(t_2) \int_0^1 h(t) f_T(t) dt \\
&=& 0.
\eeqn

%Since $f_T(t)$ is strictly decreasing (Corollary~\ref{corr:density_negative}) the mean value theorem implies
%
%\beqn
%\int_0^{t_2} h(t) f_T(t) dt > \abss{\int_{t_2}^1 h(t) f_T(t) dt }
%\eeqn

\noindent Thus the positive component of $\POW(q+1) - \POW(q)$ is larger than the negative component.

\end{proof}

%\begin{conjecture}\label{conj:power_increases_monotonically_as_rho_decreases}
%The power $\POW$ increases monotonically as $\rho$ decreases for fixed $\gamma$ and $R$.
%\end{conjecture}
%
%
%%\begin{proof}
%%Using Lemma~\ref{lemm:main_support}, we show that $\gamma(y;\rho,\beta)$ decreases as $\rho$ increases for any $y$ and $\beta$.
%%
%%
%%\beqn
%%\POW(\rho_1) - \POW(\rho_2) = \expesub{T}{F_{B}(q; n_q, T)} = \int_{0}^1F_B(q, n_q, t) (f_T(t; \rho_1) - f_T(t; \rho_2)) dt.
%%\eeqn
%
%Since $R$ is fixed, $q$ is fixed and note that $F_B(q, n_q, t)$ decreases as $t$ increases. Then, it is sufficient to demonstrate that there exists a $t_0$ where for $t < t_0$, $f_T(t; \rho_1) > f_T(t; \rho_2)$ and for $t > t_0$, $f_T(t; \rho_1) < f_T(t; \rho_2)$ for $\rho_1 < \rho_2$. From numerical integration for a fine grid of $t$ for many pairs of $\rho_1 < \rho_2$ with a pronounced difference, $t_0$ is apparent.% but for f $\rho_1 \approx \rho_2$ the crossing point is less apparent.
%%\end{proof}

\begin{theorem}\label{thm:se_pow_decreases}
Using Equation~\ref{eq:qz}, the large sample approximation of $q(z)$, the variability in the power, $\sesub{B_z}{\POW_{\z, \allocspaceR}}$, decreases monotonically (as the attainable $R$ increases) to a nonzero constant.
\end{theorem}

\begin{proof}

For a given $\z$, let $V_i$ denote the treatment effect estimator for the $i$th allocation and $V_i^m$ for its mirror defined as in Lemma~\ref{lemm:main_support},

\beqn
V_i &=& \sqrt{\rho}Z_0+\sqrt{1-\rho}Z_i+\gamma, \\
V^m_i &=& -\sqrt{\rho}Z_0-\sqrt{1-\rho}Z_i+\gamma
\eeqn

\noindent and let $I_i$ be the indicator that $V_i$ beats the $1 - \alpha$ quantile from the other $2R - 1$ estimators (and $I_i^m$ be the indicator that its mirror beats this quantile). The power conditional on $\z$ (Equation~\ref{eq:pow}) is

\beqn
\POW_{\z, \allocspaceR} = \oneover{2R} \sum_{i=1}^R (I_i + I_i^m).
\eeqn

\noindent The expected power of Equation~\ref{eq:pow_integral} is then

\bneqn\label{eq:pow_approx_for_se}
\POW = \oneover{2R} \sum_{i=1}^R (\expesub{B_Z}{I_i} + \expesub{B_Z}{I_i^m}) \approx \oneover{2R} (Rp + Rp_m) = p = p_m
\eneqn

\noindent where

\beqn
p      &:=& \int_\reals \cprob{I_i=1}{Z_0=z}\phi(z)dz,\\
p_m  &:=& \int_\reals \cprob{I_i^m=1}{Z_0=z}\phi(z)dz.
\eeqn

\noindent The approximation above in Equation~\ref{eq:pow_approx_for_se} is justified when $R$ becomes large. Under this asymptotic regime, the $1 - \alpha$ quantile converges to $q(z)$ defined in Equation~\ref{eq:qz}. Then $\expesub{B_Z}{I_i}$ becomes the integrand in the power expression of Equation~\ref{eq:asymptotic_power_integral}. The randomness in the true quantile will be a second-order effect. 

The variance of power can then be computed via

\beqn
\varsub{B_Z}{\POW_{\z, \allocspaceR}} &=& \expesub{B_Z}{\POW_{\z, \allocspaceR}^2} - \expesub{B_Z}{\POW_{\z, \allocspaceR}}^2 \\
&\approx& \frac{Rp+Rp_m+2Rp_b+4R(R-1)p_s}{4R^2} - p^2 \\
&=& \frac{\frac{1}{4}(p+p_m+2p_b) - p_s}{R}+ p_s - p^2
\eeqn

\noindent where $p_b$ and $p_s$ come from the covariance calculation:

\beqn
p_b &:=& \int_\reals \cprob{I_i=1, I_i^m=1}{Z_0=z} \phi(z) dz,\\
p_s &:=& \oneover{4} \int_\reals \squared{\cprob{I_i=1}{Z_0=z}+\cprob{I_i^m=1}{Z_0=z}} \phi(z)dz.
\eeqn

%By symmetry each ($I_i,J_i)$ has the same behavior. When that allocation generates the observations,  the value for allocation $i$ is 
%
%and when the mirror of allocation $i$ is the observed allocation then its value is

%The analysis conditions on $Z_0=z$. Denote the  respective $1-\alpha$ quantile for the assignment and its mirror, conditionally on $Z_0$,  asymtotically \textcolor{red}{I don't understand what do you mean by ``asymptotically''; Isn't the argument exact?} over other  pairs  by  $q(z)$ and $q^m(z)$ \textcolor{red}{Where is this notation used?}
%Define the following terms none of which depend on $R$
%
%\beqn
%p&=&\int_z P(I_i=1|Z_0=z|)\phi(z)dz\\
%p_m&=&\int_z P(J_i=1|Z_0=z|)\phi(z)dz\\
%p_b&=&\int_z P(I_i=1 \cap J_i=1|Z_0=z|)\phi(z)dz\\
%p_s&=&\frac{1}{4}\int_z (P(I_i=1|Z=z)+P(J_i=1|Z=z))^2\phi(z)dz
%\eeqn
%
%Note that $p=p_m=E(Power)$ asymptotically \textcolor{red}{(asymptotically?; If $p=p_m$ why do you need different notation?)}. 
%Since 
%
%\beqn
%E(Power^2)=\frac{Rp+Rp_m+2Rp_b+4R(R-1)p_s}{4R^2}
%\eeqn
%
%Hence
%
%\beqn
%V(Power)=\frac{\frac{1}{4}(p+p_m+2p_b)-p_s}{R}+p_s-p^2%=\frac{a(\rho,\beta)}{R}+b(\rho,\beta)
%\eeqn

\noindent Since $p, p_m, p_s$ and $p_b$ are constants in $R$, this proves that the variance of the power monotonically decreases in $R$ to the positive constant $p_s - p^2$. 
\end{proof}

\pagebreak
\section{Additional Figures}

\begin{figure}[htp]
\centering
\includegraphics[width=6.6in]{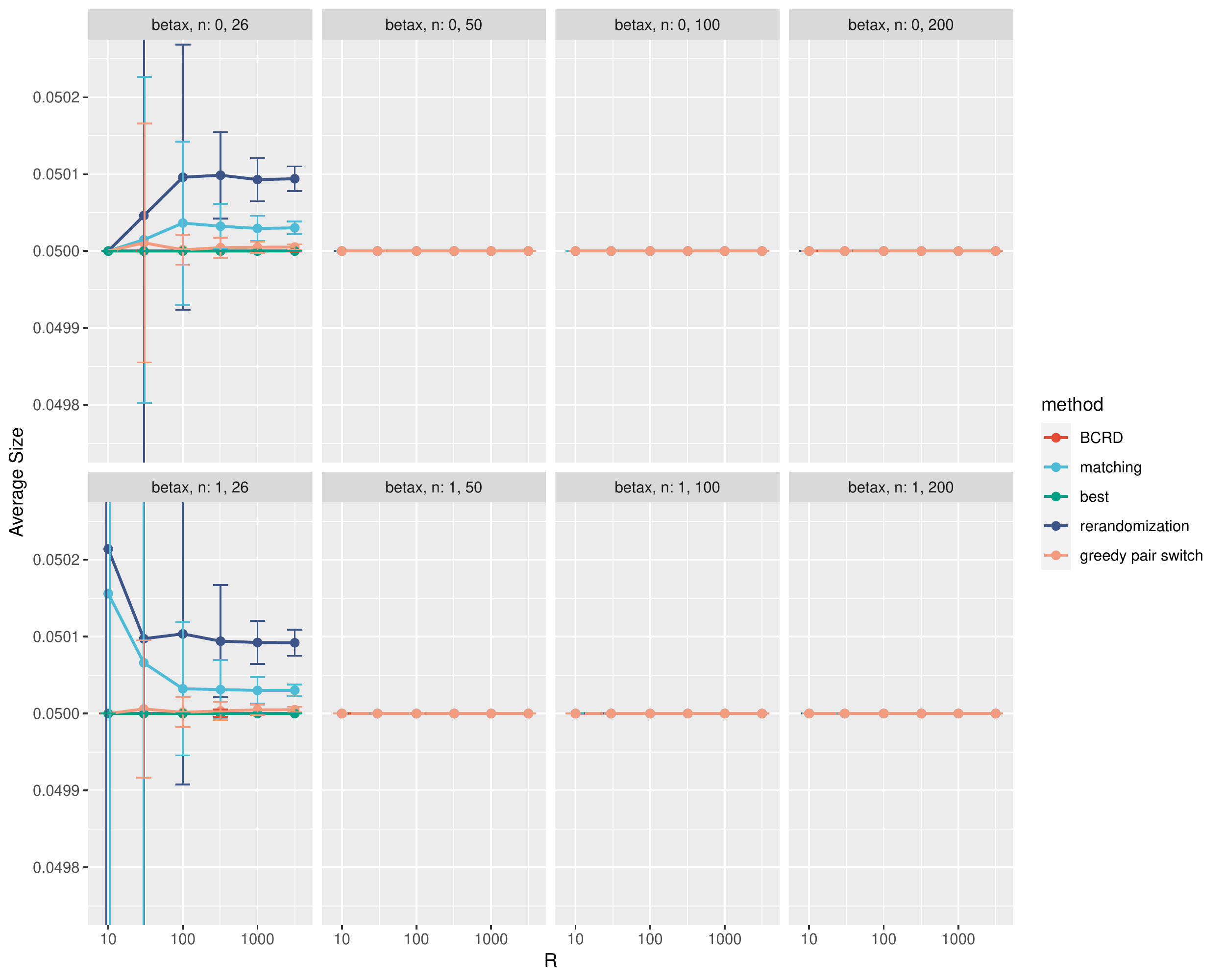}
\caption{Analogous to Figure~\ref{fig:average_power_sims} in the main text, this is the simulated size of the randomization test by number of allocation vectors in the design $R$ where $\beta = 0.25$, $\alpha = 0.05$ and $\beta_x \in \braces{0, 1}$. Individual plots correspond to different settings of the effect of the observed covariate $\beta_x$ and sample size $n$. Colors indicate the design strategy employed. Error bars are jittered slightly left-right and indicate 95\% confidence.}
\label{fig:size_by_method}
\end{figure}

\begin{figure}[h]
\centering
\includegraphics[width=6.6in]{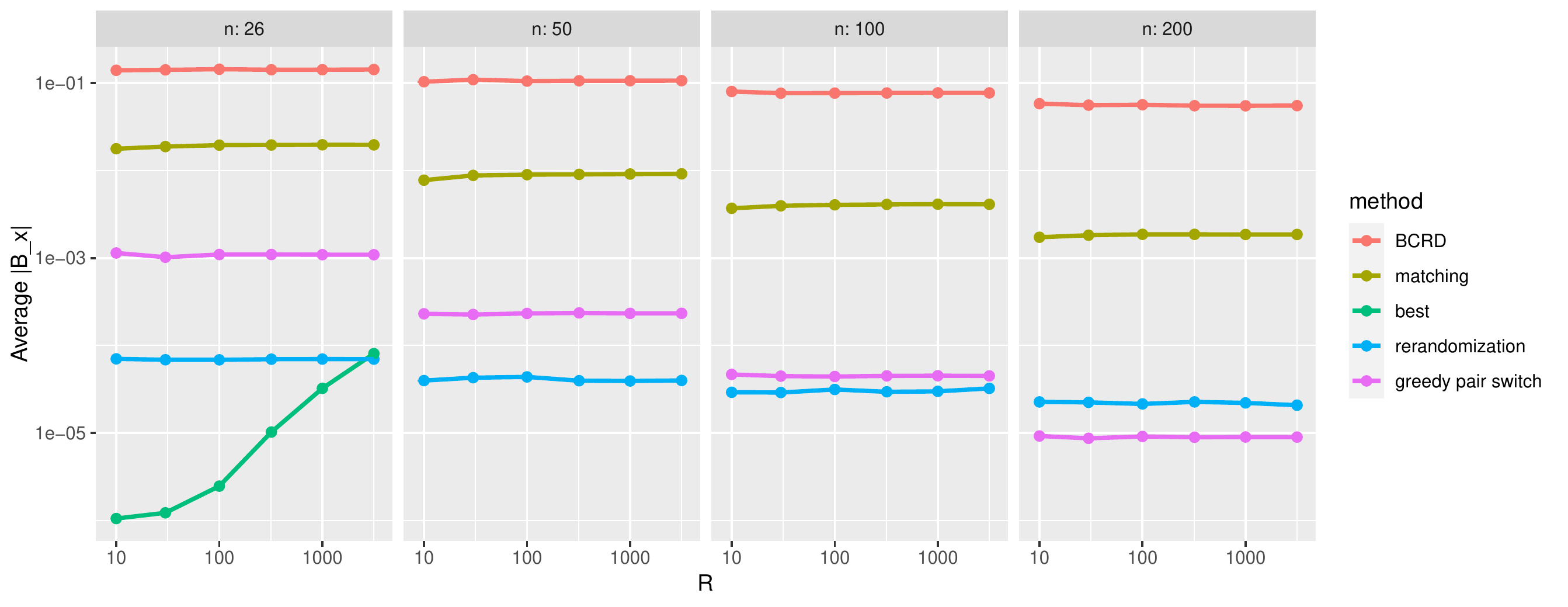}
\caption{Simulated $\abss{B_x}$ on the log scale for each design method considered by subset size $R$. Individual plots correspond to different settings the sample size $n$. Imbalance in $\x$ is independent of $\betaT, \beta_x$ and $\alpha$. Colors indicate the design strategy employed. Error bars are jittered slightly left-right and indicate 95\% confidence but are usually smaller than the dot width. One may have expectedthat as $R$ gets larger, the $B_x$ may increase appreciably lowering power. However, since the allocations spaces $\allocspace_D$ in these methods is exponentially large, we would need much larger $R$ than simulated to see this effect. This phenomenon would not be present for BCRD as all $\w$ have the same observed imbalance on average nor rerandomization as all $\w$ were selected to maintain a certain minimum $\abss{B_x}$.}
\label{fig:average_abs_Bxs}
\end{figure}

%\begin{figure}[h]
%\centering
%\includegraphics[width=6.6in]{average_abs_rijs}
%\caption{Simulated average $\abss{r_{ij}}$ over all $2R$ allocation vectors on the log scale for each design method considered by subset size $R$. Individual plots correspond to different settings the sample size $n$. Allocation vector correlations are independent of $\betaT, \beta_x$ and $\alpha$. Colors indicate the design strategy employed. Error bars are jittered slightly left-right and indicate 95\% confidence but are usually smaller than the dot width. One may have expected that as $R$ gets larger, the average $\abss{r_{ij}}$ may increase appreciably lowering power. However, this is not the case in the $R$ values we examined. The advantage of pairwise matching is clearly illustrated as well as the equivalence of BCRD, rerandomization at the threshold considered and greedy pair switching.}
%\label{fig:average_abs_rijs}
%\end{figure}

\begin{figure}[h]
\centering
\includegraphics[width=7in]{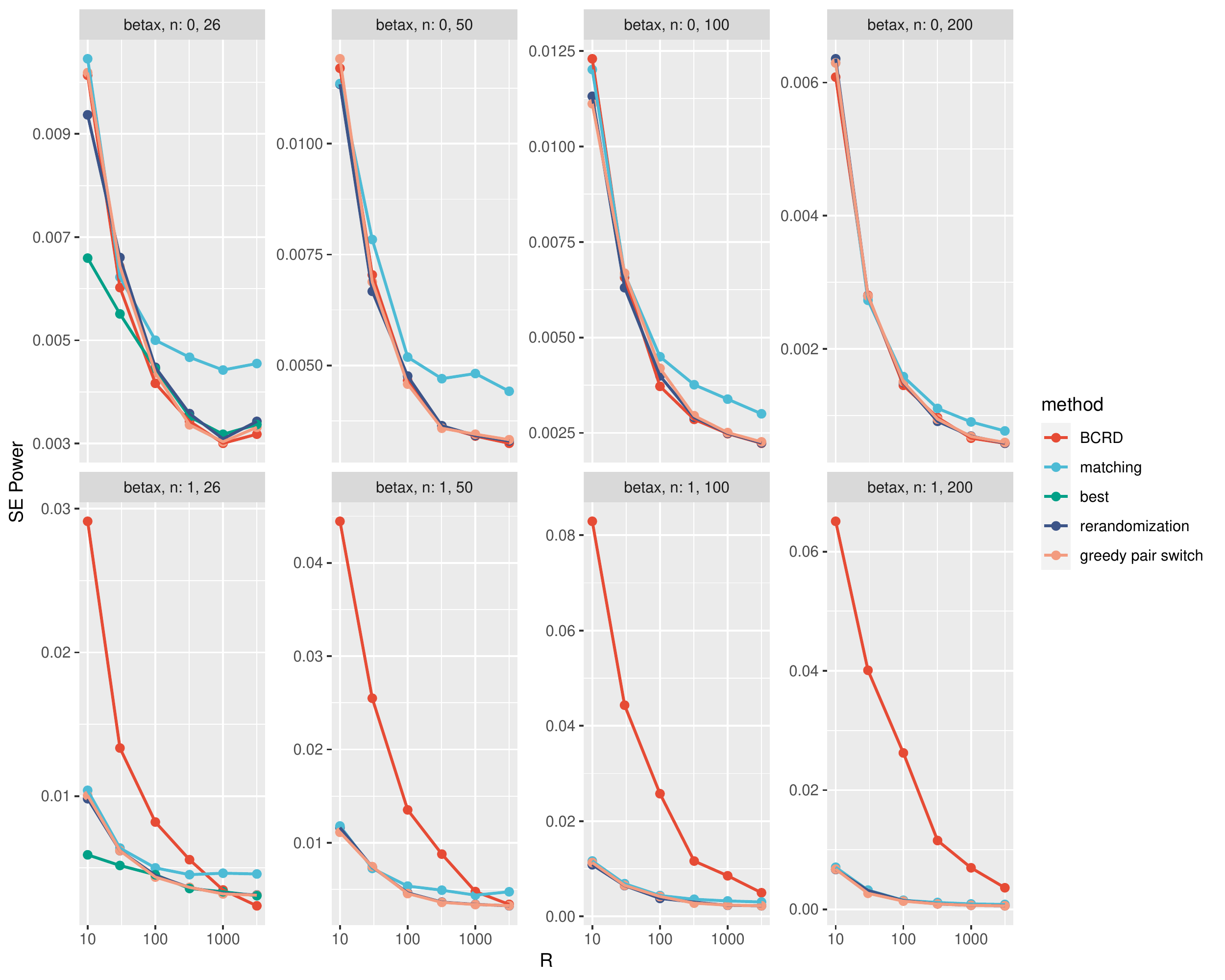}
\caption{Analogous to Figure~\ref{fig:average_power_sims} in the main text, this is the standard error of the power $\POW$ of the randomization test by number of allocation vectors in the design $R$ where $\beta = 0.25$, $\alpha = 0.05$ and $\beta_x \in \braces{0, 1}$. Individual plots correspond to different settings of the effect of the observed covariate $\beta_x$ and sample size $n$. Colors indicate the design strategy employed.}
\label{fig:ses_by_method}
\end{figure}

\end{document}